\documentclass[12pt,journal,epsfig,onecolumn,draftclsnofoot]{IEEEtran}
\usepackage{amsmath}
\usepackage{graphicx} 
\usepackage{amssymb}
\usepackage{amsfonts}
\usepackage{amsthm}
\usepackage{cite}
\usepackage{bm,comment}
\usepackage{algorithm}
\usepackage{algpseudocode}

\usepackage{subeqnarray}
\usepackage{subfigure}
\usepackage{multirow}
\usepackage{diagbox}
\usepackage{slashbox}
\usepackage{stfloats}
\usepackage{float}
\usepackage{color} 
\usepackage{cases}

\newtheorem{proposition}{Proposition}
\newtheorem{remark}{Remark}

\newtheorem{corollary}{Corollary}

\floatname{algorithm}{Algorithm}

\begin{document}
\title{Joint Power Allocation and Passive Beamforming Design for IRS-Assisted Physical-Layer Service Integration}
\author{Boyu Ning, \IEEEmembership{Student Member, IEEE}, Zhi Chen, \IEEEmembership{Senior Member, IEEE},\\
Zhongbao Tian, \IEEEmembership{Student Member, IEEE}, and Shaoqian Li, \IEEEmembership{Fellow, IEEE} \vspace{-6pt}
\thanks{This work was supported in part by the National Key R$\&$D Program of China under Grant 2018YFB1801500.}
\thanks{The authors are with National Key Laboratory of Science and Technology on Communications, University of Electronic Science and Technology of China, Chengdu 611731, China (e-mails: boydning@outlook.com; chenzhi@uestc.edu.cn; Vincent11231@outlook.com; lsq@
uestc.edu.cn).}}
\maketitle

\begin{abstract}
Intelligent reflecting surface (IRS) has emerged as an appealing solution to enhance the wireless communication performance by reconfiguring the wireless propagation environment. In this paper, we propose to apply IRS to the physical-layer service integration (PHY-SI) system, where a single-antenna access point (AP) integrates two sorts of service messages, i.e., multicast message and confidential message, via superposition coding to serve multiple single-antenna users. Our goal is to optimize the power allocation (for transmitting different messages) at the AP and the passive beamforming at the IRS to maximize the achievable secrecy rate region. To this end, we formulate this problem as a bi-objective optimization problem, which is shown equivalent to a secrecy rate maximization problem subject to the constraints on the quality of multicast service. Due to the non-convexity of this problem, we propose two customized algorithms to obtain its high-quality suboptimal solutions, thereby approximately characterizing the secrecy rate region. The resulting performance gap with the globally optimal solution is analyzed. Furthermore, we provide theoretical analysis to unveil the impact of IRS beamforming on the performance of PHY-SI. Numerical results demonstrate the advantages of leveraging IRS in improving the performance of PHY-SI and also validate our theoretical analysis.

\end{abstract}
\IEEEpeerreviewmaketitle


\section{Introduction}
Recently, intelligent reflecting surface (IRS) has emerged as an appealing candidate for future wireless communications due to its low hardware cost and energy consumption\cite{csa,cometa,basar}. Specifically, IRS is a planar array consisting of massive reflecting elements, each being able to passively reflect the incident wireless signal by adjusting its phase shift (PS). By controlling the PSs at IRS, the reflected signals of different elements can be added/counteracted in intended/unintended directions, thus resulting in a programmable and controllable wireless environment. In view of this advantage, both academia and industry have been exploring the performance optimization in different IRS-assisted systems. For example, in a multi-input single-output (MISO) IRS-assisted system, different optimization targets, e.g., transmit power minimization\cite{qinte}, weighted sum-rate maximization\cite{ghy}, and energy efficiency maximization\cite{huangchi}, etc., have been investigated by jointly optimizing the active beamforming at the access point (AP) and PSs (or passive beamforming) at IRS. In the more general and challenging multi-input multi-output (MIMO) IRS-assisted system, the authors in \cite{nby,szhang,pwang} have proposed some suboptimal algorithms, e.g., sum-path-gain maximization approach\cite{nby}, alternating optimization (AO)-based method\cite{szhang}, and manifold optimization (MO)-based algorithm\cite{pwang}, etc., to maximize the channel capacity. As a further advance, the authors in \cite{cpan} focused on a multi-cell IRS-assisted MIMO network and proposed a block coordinate descent (BCD) algorithm to maximize the weighted sum-rate of all users\cite{cpan}. 

On the other hand, physical-layer security (PLS) is also a critical issue in wireless communications. Exploiting IRS to improve the performance of PLS has been widely studied in the literature recently. For example, in the case with one legitimate user and one eavesdropper, \cite{nby2,mcui,hshen,jchen,xguan}  have proposed different algorithms, e.g., the alternating direction method of multipliers (ADMM) algorithm\cite{nby2}, semidefinite relaxation (SDR)\cite{mcui}, and majorization-minimization (MM) method\cite{hshen}, to jointly optimize the active and passive beamforming for maximizing the minimum-secrecy-rate. Furthermore, the authors in \cite{jchen} considered a more general broadcast system with multiple legitimate users and multiple eavesdroppers and proposed an alternating optimization (with path-following procedures) algorithm to maximize the secrecy rate.  Besides, it was revealed in \cite{xguan} that emitting artificial noise (AN) is an effective means for enhancing the secrecy rate in IRS-assisted systems. It is worth noting that in \cite{nby2,mcui,hshen,jchen,xguan},
it is assumed that each legitimate user only requests the confidential message while other users who do not request this message are treated as eavesdroppers.

Nevertheless, in practice, all of these users may order assorted services at the same time. For example, both confidential and public (e.g., multicast) messages may coexist by applying the emerging physical-layer service integration (PHY-SI) technique, which integrates confidential and multicast messages via superposition coding for one-time transmission, so as to improve the spectral efficiency. In fact, the concept of PHY-SI can be traced back to Csisz{\'a}r and K\"{o}rner's seminal work \cite{cjk} for a discrete memoryless broadcast channel. This work was later extended to the general MIMO systems\cite{hdl} and bidirectional relay networks\cite{rwy}. In PHY-SI, each user first decodes the common multicast message and the legitimate user further decodes the confidential message by subtracting the decoded multicast message. As such, a basic problem in PHY-SI lies in how to reconcile the trade-off between maximizing the quality of the confidential service (or secrecy rate) and that of the multicast service (or the minimum achievable rate among all users when decoding the multicast message, referred to as multicast rate in the sequel of this paper)\cite{meiac}.  This can be formulated as a bi-objective secrecy rate region maximization problem. To solve this problem, the authors in \cite{hdl} proposed a re-parameterization method to find all of its Pareto optimal solutions, i.e., the boundary points of the secrecy rate region. An alternative approach is by reformulating this bi-objective optimization problem as a secrecy rate maximization problem yet subject to an additional constraint on the quality of multicast service (QoMS). By traversing the desired multicast rate included in the QoMS constraint, all Pareto optimal solutions can be derived\cite{mei1}. Moreover, the authors in \cite{mei2} revealed that emitting AN may not always bring performance gain in PHY-SI as it results in co-channel interference to the multicast message, which is fundamentally different from PLS. Along with the advent of IRS, it is an interesting problem to study whether IRS helps enhance the performance of PHY-SI. However, to the best of our knowledge, this problem has not been studied yet. 

To fill in this gap, we consider this secrecy rate region maximization problem in a multi-user single-input single-output (SISO) IRS-assisted system, where all users have ordered the multicast service, while one user further orders the confidential service. Thus, in decoding its confidential message, all other users are treated as eavesdroppers. It is worth noting that as compared to the conventional PHY-SI without IRS, maximizing the secrecy rate region is a much more challenging problem even in the SISO case. This is due to the non-convex unit-modulus constraints on the IRS passive beamforming, which is coupled in the secrecy rate and multicast rate.  Our goal is to optimize the power allocation (for transmitting the multicast and confidential messages) at the AP and the IRS passive beamforming to find all Pareto optimal solutions. To this end, we first consider a two-user system (thus with only a single eavesdropper). Specifically, we show that the bi-objective optimization problem is equivalent to a secrecy rate maximization problem subject to QoMS constraints. However, this equivalent problem is still non-convex. To tackle this challenge, we propose a Charnes-Cooper transformation (CCT)-based algorithm to approximately solve it via the SDR technique. Furthermore, we show that the optimal power allocation ratio can be obtained in closed-form with any given passive beamforming. By leveraging this fact, we also proposed a weighted sum covariance matrix (WSCM)-based algorithm with lower complexity to solve the considered problem. Next, we extend the above algorithms to accommodate the general multi-user/eavesdropper scenario and characterize the performance gap to the globally optimal solution. To draw useful insights, we perform theoretical analysis  to unveil the impact of IRS passive beamforming on the performance of PHY-SI versus that without IRS. Numerical results demonstrate that the performance of our proposed algorithms can yield near-optimal performance. The results also reveal that the secrecy rate region can be significantly enlarged by leveraging IRS compared to the scheme without IRS. In particular, by equipping the IRS with sufficient number of elements, the number of users who can achieve positive secrecy rate can be increased significantly. 

The rest of this paper is organized as follows. Section II introduces the system model and the problem formulation. In Sections III and IV, efficient algorithms are proposed to solve the formulated problems in the two-user and multi-user cases, respectively.
Section V provides the theoretical analysis on the IRS-assisted PHY-SI. Section V presents the numerical results to evaluate the efficacy of our proposed algorithms. Finally, we conclude the paper in Section VI.

\indent \emph{Notation:} We use small normal face for scalars, small bold face for vectors, and capital bold face for matrices. $\rm{Tr}(\cdot)$, $\rm{rank}(\cdot)$, and $|\cdot|$ represent the trace, rank, and modulus of a matrix, respectively. The superscript ${{\rm{\{ }} \cdot \}^*}$ and ${{\rm{\{ }} \cdot \}^H}$ denote the conjugate and Hermitian transpose. ${\bf{a}}(n)$ represents the $n$th element of ${\bf{a}}$. ${\rm{diag}}( \cdot )$ denotes a diagonal matrix whose diagonal elements are given by its argument. $\mathcal{CN}(\mu,\sigma^2)$ means circularly symmetric complex Gaussian (CSCG) distribution with mean of $\mu$ and variance of $\sigma^2$. ${{\bf{I}}}$ denotes the identity matrix. $\mathbb{K}$ represents a proper cone, and $\mathbb{K}^*$ represents a dual cone associated with $\mathbb{K}$.

\begin{figure}[t]
\centering
\includegraphics[width=3in]{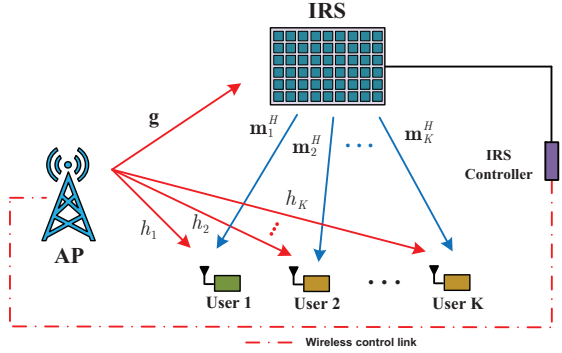}
\caption{An IRS-assisted PHY-SI system.}\label{mo}
\vspace{-6pt}
\end{figure}

\section{System Model and Problem Formulation}
We consider an IRS-assisted communication system as depicted in Fig. \ref{mo}, where a single-antenna AP serves $K$ single-antenna users, with the assistance of an IRS equipped with $N$ passive PSs and installed on a nearby wall. We assume that all the users have requested a multicast service and user 1 further requests a confidential service. As such, the other users serve as potential internal eavesdroppers which are able to eavesdrop the confidential information intended for user 1. For ease of exposition, we define $\mathcal{K}\triangleq \{1,2,...,K\}$ and $\mathcal{E}\triangleq \mathcal{K}/\{1\}$ to represent the index set of all users and that of eavesdroppers, respectively. In PHY-SI, the AP transmits the superposition of the confidential message $s_c \sim \mathcal{CN}(0,1)$ and the multicast message $s_m \sim \mathcal{CN}(0,1)$, denoted as
\begin{equation}
{x = }\sqrt{\alpha} {s_c} + \sqrt{\beta} {s_m},
\end{equation}
where $\alpha $ and $\beta$ represent the transmit power allocated to the confidential message and multicast message, respectively. To characterize the theoretical limit of the considered system, we assume that the channel state information (CSI) on all links involved are available at the AP and can be sent to the IRS controller for tuning the IRS's PSs. In practice, this channel knowledge can be obtained by using some customized channel estimation schemes for IRS-aided wireless communications (see \cite{est1,est2,est3}). In addition, we consider the quasi-static block fading channels, i.e., all channels involved remain constant within the considered fading block. Let $\{h_k\}_{k=1}^K \in {\mathbb{C}}$, $\{{\bf{m}}_k^H\}_{k=1}^K \in {\mathbb{C}^{1 \times N}}$, and ${\bf{g}} \in {\mathbb{C}^{{N} \times 1}}$ denote the baseband equivalent channels of AP-user$k$ link, IRS-user$k$ link, and AP-IRS link, respectively. With the IRS, the transmitted signal via ${\bf{g}} $ can be dynamically altered by its $N$ PSs, denoted as  ${\bf{\Theta }} = {\rm{diag}}( {e^{j{\theta _1}}}, {e^{j{\theta _2}}}, \cdots , {e^{j{\theta _{N}}}})$, and then reflected to the users via channels $\{{\bf{m}}_k^H\}_{k=1}^K$, where $j = \sqrt {-1}$ is an imaginary unit and $\theta _i \in [0,2\pi), i=1,2,\cdots,N$ is PS in the $i$-th reflecting element of the IRS.  The overall received signal at user $k$ from both AP-user link and AP-IRS-user link can be expressed as\cite{hdl}
\begin{equation}
{y_k} = \sqrt \alpha({\bf{m}}_k^H{\bf{\Theta g}} + {h_k}){s_c} + \sqrt \beta ({\bf{m}}_k^H{\bf{\Theta g}} + {h_k}){s_m} + {n_k},\;\;k\in \mathcal{K}, 
\end{equation}
where $n_k \sim \mathcal{CN}(0,\sigma _k^2)$ is the zero-mean additive Gaussian noise with power $\sigma _k^2$. Here, we ignore the signals reflected by the IRS more than once due to the severe path-loss. Denote $R_c$ and $R_m$ as the achievable rates of the confidential message and the multicast message, respectively. Given the total power constraint $\alpha + \beta \leq P$, the secrecy rate region $C_s(\{h_k\}_{k=1}^K,\{{\bf{m}}_k^H\}_{k=1}^K,{\bf{g}},P)$ is defined as a set of non-negative rate pairs $({R_m},{R_c})$ satisfying\cite{hdl}
\begin{subequations}\label{333}
\begin{align}
&{R_m} \le \mathop {\min }\limits_{k \in \mathcal{K}} \log \left( {1 + \frac{{\beta |{\bf{m}}_k^H{\bf{\Theta g}} + {h_k}{|^2}}}{{\sigma _k^2 + \alpha |{\bf{m}}_k^H{\bf{\Theta g}} + {h_k}{|^2}}}} \right),\label{3a}\\ 
&{R_c} \le \log \left( {1 + \frac{{\alpha |{\bf{m}}_1^H{\bf{\Theta g}} + {h_1}{|^2}}}{{\sigma _1^2}}} \right) - \mathop {\max }\limits_{k \in \mathcal{E}}\log \left( {1 + \frac{{\alpha |{\bf{m}}_k^H{\bf{\Theta g}} + {h_k}{|^2}}}{{\sigma _k^2}}} \right).\label{3b}
\end{align}
\end{subequations}
It is observed from (\ref{333}) that all users first decode the multicast message by treating the confidential message as noise. By removing the decoded multicast message, user 1 then decodes the confidential message without interference from the multicast message. 

In this work, we focus on characterizing the secrecy rate region $C_s$ in the presence of an assisted IRS. This is equivalent to solving a bi-objective optimization problem below with cone $\mathbb{K} = {\mathbb{K}^*} = {\mathbb{R}}_+^2$, i.e.,
\begin{subequations}\label{4or}
\begin{align}
({\rm{P}}1):\;\mathop {\max }\limits_{\alpha ,\beta ,{\bf{v}},{R_m},{R_c}} &({\rm{w}}.{\rm{r}}{\rm{.t}}{\rm{. }}\;{\mathbb{R}}_+^2)\;({R_m},{R_c})\\
{\rm{s}}{\rm{.t}}{\rm{.}}\;\;&\mathop {\min }\limits_{k \in \mathcal{K}} \log \left( {1 + \frac{{\beta |{\bf{m}}_k^H{\rm{diag}}({\bf{v}}^*){\bf{g}} + {h_k}{|^2}}}{{\sigma _k^2 + \alpha |{\bf{m}}_k^H{\rm{diag}}({\bf{v}}^*){\bf{g}} + {h_k}{|^2}}}} \right) \ge {R_m},\label{4b}\\
&\mathop {\min }\limits_{k \in \mathcal{E}}\log \frac{{\sigma _1^2\sigma _k^2{\rm{ + }}\sigma _k^2\alpha |{\bf{m}}_1^H{\rm{diag}}({{\bf{v}}^*}){\bf{g}} + {h_1}{|^2}}}{{\sigma _1^2\sigma _k^2{\rm{ + }}\sigma _1^2\alpha |{\bf{m}}_k^H{\rm{diag}}({{\bf{v}}^*}){\bf{g}} + {h_k}{|^2}}} \ge {R_c},\label{4c}\\
&\;\alpha  + \beta  \le P,\;\;\alpha  \ge 0,\;\;\beta  \ge 0,\;\; \left| {{\bf{v}}(i)} \right| = 1, \;\; i=1,2,\cdots,N.
\end{align}
\end{subequations}
where ${\bf{v}} = [{e^{j{\theta _1}}},{e^{j{\theta _2}}}, \cdots ,e^{j{\theta _{N_r}}}]^H$ denotes the PS vector at the IRS, i.e., ${\bf{\Theta }} = {\rm{diag}}({\bf{v}}^*)$. The constraint in (\ref{4c}) is derived by reformulating (\ref{3b}). For $({\rm{P}}1)$, we have the following proposition.
\begin{proposition}
$({\rm{P}}1)$ is feasible if  ${\big| {\sum\limits_{i = 1}^N {|{{\bf{m}}_1}(i)||{\bf{g}}(i)|}  + {h_1}} \big|^2} \ge \mathop {\max }\limits_k \frac{{\sigma _1^2}}{{\sigma _k^2}}{\left| {\sum\limits_{i = 1}^N {|{{\bf{m}}_k}(i)||{\bf{g}}(i)|}  + {h_k}} \right|^2}$, and $({\rm{P}}1)$ is infeasible if there exists a $k$ satisfying ${\big| {\sum\limits_{i = 1}^N {|{{\bf{m}}_1}(i)||{\bf{g}}(i)|}  + {h_1}} \big|^2} \le \frac{{\sigma _1^2|{h_k}{|^2}}}{{\sigma _k^2}}$.
\end{proposition}
\begin{proof}
See Appendix A.
\end{proof}
Note that $({\rm{P}}1)$  is a bi-objective optimization problem. Any Pareto optimal solution to this problem is a boundary point of $C_s$. However, due to the non-convex constraints with regard to (w.r.t.) the coupled variables and the unit-modulus constraint on $\bf{v}$, $({\rm{P}}1)$  is challenging to be solved optimally. To characterize the secrecy rate region, next, we propose an approach to find all (approximate) Pareto optimal points of $({\rm{P}}1)$ .


\section{Two-User System}
In this section, we first consider a two-user system, i.e., $K=2$.

\subsection{An Equivalent Scalar Maximization Problem of $({\rm{P}}1)$ }
In order to find all the Pareto optimal solutions to $({\rm{P}}1)$, we formulate an equivalent scalar maximization problem of $({\rm{P}}1)$ by fixing one entry of  $({R_m},{R_c})$ and maximizing the other one. Specifically, we fix the variable ${R_m}$ as a constant ${r_m} \in [0,r_m^{\max}]$, where $r_m^{\max}$ is the maximum achievable multicast rate. As such, maximizing the vector $({R_m},{R_c})$ reduces to maximizing $R_c$, which is equivalent to the following problem,
\begin{subequations}
\begin{align}
({\rm{P}}2):\;{\overline R _c}({r_m}) =& \mathop {\max }\limits_{\alpha ,\beta ,{\bf{v}}} \;{\left\lceil {\log \frac{{\sigma _1^2\sigma _2^2{\rm{ + }}\sigma _2^2\alpha |{\bf{m}}_1^H{\rm{diag}}({{\bf{v}}^*}){\bf{g}} + {h_1}{|^2}}}{{\sigma _1^2\sigma _2^2{\rm{ + }}\sigma _1^2\alpha |{\bf{m}}_2^H{\rm{diag}}({{\bf{v}}^*}){\bf{g}} + {h_2}{|^2}}}} \right\rceil ^ + }\\
&{\rm{s}}{\rm{.t}}{\rm{.}}\;\ \log \left( {1 + \frac{{\beta |{\bf{m}}_1^H{{\rm{diag}}({\bf{v}}^*)}{\bf{g}} + {h_1}{|^2}}}{{\sigma _1^2 + \alpha |{\bf{m}}_1^H{{\rm{diag}}({\bf{v}}^*)}{\bf{g}} + {h_1}{|^2}}}} \right) \ge {r_m},\label{5b}\\
&\qquad \log \left( {1 + \frac{{\beta |{\bf{m}}_2^H{{\rm{diag}}({\bf{v}}^*)}{\bf{g}} + {h_2}{|^2}}}{{\sigma _2^2 + \alpha |{\bf{m}}_2^H{{\rm{diag}}({\bf{v}}^*)}{\bf{g}} + {h_2}{|^2}}}} \right) \ge {r_m},\label{5c}\\
&\qquad \alpha  + \beta  \le P,\;\;\alpha  \ge 0,\;\;\beta  \ge 0,\;\;\left| {{\bf{v}}(i)} \right| = 1,\;\;i = 1,2, \cdots ,N.
\end{align}
\end{subequations}
$({\rm{P}}2)$ can be interpreted as finding the maximum achievable secrecy rate under the QoMS constraint. By following the Theorem 1 given in \cite{mei1}, it is easy to prove that $(r_m,{\overline R _c}({r_m}))$ must be a Pareto optimal solution to $({\rm{P}}1)$, i.e., a boundary point of $C_s$. In addition, some interesting properties w.r.t. $({\rm{P}}2)$ are listed as follows.
\begin{itemize}
\item When $r_m=0$, the QoMS constraint is inactive and $({\rm{P}}2)$  becomes a conventional secrecy rate maximization problem with ${\overline R _c}(0)=r_c^{\max}$, where $r_c^{\max}$ is the maximum achievable secrecy rate in our considered system.
\item When $r_m$ increases, the feasible region of $({\rm{P}}2)$  will be shrank. Thus, the optimum ${\overline R _c}({r_m})$ is monotonically non-increasing with $r_m$.
\item When $r_m=r_m^{\max}$, all transmit power must be allocated to transmit the multicast message for achieving the maximum multicast rate, i.e., $\beta=P$ and $\alpha=0$, with ${\overline R _c}(r_m^{\max})=0$.
\end{itemize}
It follows from the above properties that with increasing $r_m$ from $0$ to $r_m^{\max}$, the optimal value ${\overline R _c}({r_m})$ of (P2) will decrease from $r_c^{\max}$ to $0$. Thus, the two endpoints of the region $C_s$ are given by $(0,r_c^{\max})$ and $(r_m^{\max},0)$, respectively. As a result, by varying the parameter ${r_m}$ within $[0,r_m^{\max}]$, all the Pareto optimal solutions to (P1) can be obtained by solving (P2).


\subsection{An Upper Bound on the Maximum Multicast Rate}
However, it is generally difficult to obtain $r_m^{\max}$ since maximizing the multicast rate is also a challenging optimization problem, i.e., 
\begin{equation}
\begin{split}
({\rm{P}}3):r_m^{\max } = \mathop {\max }\limits_{\beta ,{\bf{v}}} \mathop {\min }\limits_{k = 1,2} \log \left( {1 + \frac{\beta }{{\sigma _k^2}}|{\bf{m}}_k^H{\rm{diag}}({{\bf{v}}^*}){\bf{g}} + {h_k}{|^2}} \right)\\
{\rm{s}}{\rm{.t}}{\rm{.}}\;\;\beta  \le P,\;\beta  \ge 0,\;\;\left| {{\bf{v}}(i)} \right| = 1,\;\;i = 1,2, \cdots ,N.
\end{split}
\end{equation}
Thus, we aim to find an upper bound on $r_m^{\rm{max}}$, denoted as $r_m^{\rm{up}}$, by solving $({\rm{P}}3)$ via the SDR technique. By varying the parameter ${r_m}$ within $[0,r_m^{\rm{up}}]$, we can still find all the Pareto optimal points by solving $({\rm{P}}2)$. In particular, when ${r_m} \in [r_m^{\max},r_m^{\rm{up}}]$, there is no solution to $({\rm{P}}2)$. To reformulate $({\rm{P}}3)$  into an SDR problem, we first rewrite its objective function as an equivalent form. Specifically, by using the equality ${\bf{m}}_k^H{\rm{diag}}({{\bf{v}}^*}) = {{\bf{v}}^H}{\rm{diag}}({\bf{m}}_k^*)$, we have
\begin{equation}\label{fom}
\begin{split}
&\log \left( {1 + \frac{\beta }{{\sigma _k^2}}|{\bf{m}}_k^H{\rm{diag}}({{\bf{v}}^*}){\bf{g}} + {h_k}{|^2}} \right)\\
&\;\; = \log \Big[ {1 + \frac{\beta }{{\sigma _k^2}}\left( {{{\bf{v}}^H}{\rm{diag}}({\bf{m}}_k^*){\bf{g}}{{\bf{g}}^H}{\rm{diag}}({\bf{m}}_k^{}){\bf{v}}} \right.} \\
&\;{\left. {\;\;\;\;\;\;\;\;\;\qquad + {{\bf{v}}^H}{\rm{diag}}({\bf{m}}_k^*){\bf{g}}h_k^* + {h_k}{{\bf{g}}^H}{\rm{diag}}({\bf{m}}_k^{}){\bf{v}} + h_k^*{h_k}} \right)} \Big]\\
&\;\; = \log \left( {1 + \frac{\beta }{{\sigma _k^2}}{{\bf{z}}^H}{{\bf{T}}_k}{\bf{z}}} \right)
\end{split}
\end{equation}
where ${\bf{z}} = [t \cdot {{\bf{v}}^H},\;t]^H$ and $t$ is an introduced auxiliary variable satisfying $|t|=1$. The matrix ${{{\bf{T}}_k}}$ is given by 
\begin{equation}
{{\bf{T}}_k} = \left[ {\begin{array}{*{20}{c}}
{{\rm{diag}}({\bf{m}}_k^*){\bf{g}}{{\bf{g}}^H}{\rm{diag}}({\bf{m}}_k^{})}&{{\rm{diag}}({\bf{m}}_k^*){\bf{g}}h_k^*}\\
{{h_k}{{\bf{g}}^H}{\rm{diag}}({\bf{m}}_k^{})}&{h_k^*{h_k}}
\end{array}} \right],\;\;k = 1,2.
\end{equation}
By noting that ${{\bf{z}}^H}{{\bf{T}}_k}{\bf{z}}{\rm{ = Tr(}}{\bf{z}}{{\bf{z}}^H}{{\bf{T}}_k}{\rm{)}}$, we define ${\bf{Z}} = {\bf{z}}{{\bf{z}}^H}$ which satisfies ${\bf{Z}} \succeq 0$ and ${\rm{rank}}({\bf{Z}}) = 1$. It is obvious that the maximum multicast rate is achieved when setting $\beta=P$. By relaxing the rank one constraint and introducing a slack variable $\tau$, $({\rm{P}}3)$ becomes a convex optimization problem, i.e.,
\begin{equation}
\begin{split}
{\rm{(\hat P3):}}\;\;&r_m^{\rm{up}}=\mathop {\max }\limits_{{\bf{Z}},\tau } \;\tau \\
&{\rm{s}}{\rm{.t}}{\rm{.}}\;\;{\rm{Tr}}(\frac{P}{{\sigma _1^2}}{\bf{Z}}{{\bf{T}}_1}) \ge {2^\tau } - 1, \;\;{\rm{Tr}}(\frac{P }{{\sigma _2^2}}{\bf{Z}}{{\bf{T}}_2}) \ge {2^\tau } - 1, \\
&\qquad\;{\rm{Tr}}({\bf{ZE}}_i)=1,\;\;i = 1,2,...,N+1,\;\;{\bf{Z}} \succeq {\bf{0}}.
\end{split}
\end{equation}
The relaxed problem $(\hat{\rm{P}}3)$ is a convex semidefinite program (SDP), thus can be optimally solved by the interior-point method or CVX\cite{cvx}. Generally, due to the relaxed feasible region, the solution to $(\hat{\rm{P}}3)$ may not be rank one and thus we have $r_m^{\rm{up}}\geq r_m^{\max}$.

\subsection{Proposed Solution to $({\rm{P}}2)$}
Given $r_m \in [0,r_m^{\rm{up}}]$, we aim to find all Pareto optimal solutions to $({\rm{P}}1)$ by solving $({\rm{P}}2)$. However, $({\rm{P}}2)$  is a nonconvex optimization problem. In the sequel, we propose a CCT-based algorithm to tackle it for characterizing the secrecy rate region. 
\begin{remark}
When ${\overline R _c}(0)=0$, the secrecy rate region $C_s$ reduces to a line segment on the axis of the multicast rate, i.e., ${C_s} = \left\{ {(0,{R_m})|{R_m} \in [0,r_m^{\max}]} \right\}$. In this case, our considered problem is degraded into the multicast rate maximization problem, i.e., $({\rm{P}}3)$.
\end{remark}
Remark 1 implies that $({\rm{P}}2)$  is trivial when ${\overline R _c}(0)=0$, thus we assume ${\overline R _c}(0)>0$  in this paper, which is equivalent to the following condition,
\begin{equation}\label{cond}
\frac{{|{\bf{m}}_1^H{\rm{diag}}({{\bf{v}}^*}){\bf{g}} + {h_1}{|^2}}}{{\sigma _1^2}} > \frac{{|{\bf{m}}_2^H{\rm{diag}}({{\bf{v}}^*}){\bf{g}} + {h_2}{|^2}}}{{\sigma _2^2}}.
\end{equation}
Define a function $f(x) = \log [1 + \beta x/(1 + \alpha x)]$ and it is easy to verify that $f(x)$ is monotonically increasing with $x$. Note that the left hand sides (LHSs) of (\ref{5b}) and (\ref{5c}) can be expressed as $f(|{\bf{m}}_1^H{\rm{diag}}({{\bf{v}}^*}){\bf{g}} + {h_1}{|^2}/\sigma _1^2)$ and $f(|{\bf{m}}_2^H{\rm{diag}}({{\bf{v}}^*}){\bf{g}} + {h_2}{|^2}/\sigma _2^2)$, respectively. Thus, (\ref{cond}) implies that the former must be greater than the latter. Hence, constraint (\ref{5b}) is redundant and can be discarded in the optimization. Next, we further simplify $({\rm{P}}2)$ based on the facts below.
\begin{enumerate}
\item We ignore the logarithmic function due to the monotonicity property.
 \item The total transmit power constraint is tight when the maximum secrecy rate is achieved, i.e., $\alpha + \beta  = P$.
 \item Similar to the reformulation in (\ref{fom}), we have the following equivalent relations,
\begin{align}
&\frac{{\sigma _1^2\sigma _2^2{\rm{ + }}\sigma _2^2\alpha |{\bf{m}}_1^H{\rm{diag}}({{\bf{v}}^*}){\bf{g}} + {h_1}{|^2}}}{{\sigma _1^2\sigma _2^2{\rm{ + }}\sigma _1^2\alpha |{\bf{m}}_2^H{\rm{diag}}({{\bf{v}}^*}){\bf{g}} + {h_2}{|^2}}} = \frac{{\sigma _1^2\sigma _2^2{\rm{ + }}\sigma _2^2\alpha {\rm{Tr}}({\bf{Z}}{{\bf{T}}_1})}}{{\sigma _1^2\sigma _2^2{\rm{ + }}\sigma _1^2\alpha {\rm{Tr}}({\bf{Z}}{{\bf{T}}_2})}} = \frac{{{\rm{Tr[}}{\bf{Z}}(\frac{{\sigma _1^2\sigma _2^2}}{N+1}{\bf{I}} + \sigma _2^2\alpha {{\bf{T}}_1})]}}{{{\rm{Tr[}}{\bf{Z}}(\frac{{\sigma _1^2\sigma _2^2}}{N+1}{\bf{I}} + \sigma _1^2\alpha {{\bf{T}}_2})]}},\\
&\frac{{\beta |{\bf{m}}_2^H{\bf{\Theta g}} + {h_2}{|^2}}}{{\sigma _2^2 + \alpha |{\bf{m}}_2^H{\bf{\Theta g}} + {h_2}{|^2}}} = \frac{{(P - \alpha ){\rm{Tr}}({\bf{Z}}{{\bf{T}}_2})}}{{\sigma _2^2 + \alpha {\rm{Tr}}({\bf{Z}}{{\bf{T}}_2})}} = \frac{{{\rm{Tr[ }}(P-\alpha ){\bf{Z}}{{\bf{T}}_2}]}}{{{\rm{Tr[}}{\bf{Z}}(\frac{{\sigma _2^2}}{N+1}{\bf{I}} + \alpha {{\bf{T}}_2})]}}.
\end{align}
\end{enumerate}
Based on the above facts, solving $({\rm{P}}2)$ is equivalent to solving the following problem $({\rm{P}}4)$,
\begin{equation}
\begin{split}
({\rm{P}}4):\;{2^{{{\overline R }_c}({r_m})}} =& \mathop {\max }\limits_{\alpha ,{\bf{Z}}} \frac{{{\rm{Tr[}}{\bf{Z}}(\frac{{\sigma _1^2\sigma _2^2}}{N+1}{\bf{I}} + \sigma _2^2\alpha {{\bf{T}}_1})]}}{{{\rm{Tr[}}{\bf{Z}}(\frac{{\sigma _1^2\sigma _2^2}}{N+1}{\bf{I}} + \sigma _1^2\alpha {{\bf{T}}_2})]}}\\
&{\rm{s}}{\rm{.t}}{\rm{.}}\;\frac{{{\rm{Tr[}}(P -\alpha ){\bf{Z}}{{\bf{T}}_2}]}}{{{\rm{Tr[}}{\bf{Z}}(\frac{{\sigma _2^2}}{N+1}{\bf{I}} + \alpha {{\bf{T}}_2})]}} \ge {2^{{r_m}}} - 1,\\
&\qquad{\rm{Tr(}}{\bf{Z}}{{\bf{E}}_i}) = 1,\;\;i = 1,2,...,N+1,\;\;\;\;\quad\\
&\qquad{\rm{rank}}({\bf{Z}})=1,\;{\bf{Z}} \succeq {\bf{0}},\;\alpha  \ge 0,\;P - \alpha  \ge 0.
\end{split}
\end{equation}
Problem $({\rm{P}}4)$  is a fractional programming problem. To solve it, we propose an efficient approach to transform this problem into a more tractable from by applying the CCT, i.e.,
\begin{equation}
{\bf{Z}} = {\bf{Y}}/\xi ,\;\;\;\;\xi  > 0.
\end{equation}
Meanwhile, we relax the rank one constraint and $({\rm{P}}4)$  is reduced to
\begin{equation}
\begin{split}
{\rm{(\hat P4)}}\;\;& \mathop {\max }\limits_{\alpha ,{\bf{Y}},\xi } {\rm{Tr[}}{\bf{Y}}(\frac{{\sigma _1^2\sigma _2^2}}{N+1}{\bf{I}} + \sigma _2^2\alpha {{\bf{T}}_1})]\\
&{\rm{s}}{\rm{.t}}{\rm{.}}\;{\rm{Tr[}}{\bf{Y}}(\frac{{\sigma _1^2\sigma _2^2}}{N+1}{\bf{I}} + \sigma _1^2\alpha {{\bf{T}}_2})] = 1,\\
&\qquad{\rm{Tr}}\left\{ {{\bf{Y}}\left[ {{\rm{(P - }}\alpha {2^{{r_m}}}){{\bf{T}}_2} - \frac{{({2^{{r_m}}} - 1)\sigma _2^2}}{N+1}{\bf{I}}} \right]} \right\} \ge 0,\\
&\qquad{\rm{Tr(}}{\bf{Y}}{{\bf{E}}_i}) = \xi ,\;\;i = 1,2,...,N+1,\\
&\qquad\;{\bf{Y}} \succeq {\bf{0}},\;\alpha  \ge 0,\;P - \alpha  \ge 0,\;\xi  > 0.
\end{split}
\end{equation}
Note that ${\rm{(\hat P4)}}$ is a coordinate-wise convex optimization problem w.r.t $\alpha$ and $\{{\bf{Y}},\xi \}$. A straightforward approach to solve it is by using the alternating optimization (AO), i.e.,  optimizing $\alpha$ with fixed $\{{\bf{Y}},\xi \}$ and optimizing $\{{\bf{Y}},\xi \}$ with fixed $\alpha$ in an alternate manner. However, as the objective function of ${\rm{(\hat P4)}}$ is not jointly convex w.r.t $\{{\alpha,\bf{Y}},\xi \}$, AO generally can only yield a local optimal solution. To improve over AO, by noting $\alpha \in [0,P]$, the globally optimal solution to ${\rm{(\hat P4)}}$ can be obtained via a one dimensional search over $\alpha$ whilst optimizing $\{{\bf{Y}},\xi \}$ in every search. In particular,  ${\rm{(\hat P4)}}$ is an SDP with fixed $\alpha$ and the optimal $\alpha$ should be chosen as the one that leads to the maximum objective value of ${\rm{(\hat P4)}}$.

Generally, the globally optimal solution $\{{\alpha,\bf{Y}},\xi \}$ to ${\rm{(\hat P4)}}$ may not lead to a feasible PS vector ${\bf{v}}$ to ${\rm{(P2)}}$ unless the rank one constraint is satisfied, i.e., ${\rm{rank(}}{\bf{Y}}/\xi ) = 1$. If not, we can perform the well-known Gaussian randomization procedure (GRP) to derive a feasible PS vector. Specifically, we carry out the eigenvalue decomposition of ${\bf{Z}}={\bf{Y}}/\xi $, denoted as ${\bf{Z}}= {\bf{U\Sigma }}{{\bf{U}}^H}$, and construct a random vector $\widetilde {\bf{z}} = {\bf{U}}{{\bf{\Sigma }}^{1/2}}{\bf{r}}$, where ${\bf{r}}\in {\mathbb{C}}^{N+1}$ is a random vector following $\mathcal{CN}({\bf{0}},{\bf{I}})$. Then we can obtain a feasible PS vector as
\begin{equation}\label{grp}
{\bf{v}} = \left[ {\frac{{\widetilde {\bf{z}}(1)/\widetilde {\bf{z}}(N + 1)}}{{\left| {\widetilde {\bf{z}}(1)/\widetilde {\bf{z}}(N + 1)} \right|}},\frac{{\widetilde {\bf{z}}(2)/\widetilde {\bf{z}}(N + 1)}}{{\left| {\widetilde {\bf{z}}(2)/\widetilde {\bf{z}}(N + 1)} \right|}},...,\frac{{\widetilde {\bf{z}}(N)/\widetilde {\bf{z}}(N + 1)}}{{\left| {\widetilde {\bf{z}}(N)/\widetilde {\bf{z}}(N + 1)} \right|}}} \right].
\end{equation}
By this means, ${\bf{v}}$ is determined as the one which attains the maximum objective value of ${\rm{(P2)}}$ among a set of randomly generated PS vectors via GRP.

\subsection{Lower-Complexity Algorithm}
In the last subsection, we have proposed a CCT-based algorithm to characterize the secrecy rate region in the IRS-assisted PHY-SI. However, in this algorithm, we need to solve an SDP with the path-following approach at every point in the linear search, thus incurring high computational complexity. In this subsection, we propose a WSCM-based algorithm as an suboptimal alternative but with lower computational complexity. Specifically, we first derive an optimal covariance matrix ${{\bf{Z}}_{m\;}}$ by solving $(\hat{\rm{P}}3)$ for maximizing the multicast rate. Then we derive an optimal covariance matrix ${{\bf{Z}}_{c\;}}$  for maximizing the secrecy rate, which can be obtained by solving the following SDP and setting ${{\bf{Z}}_c} = {{\bf{Y}}_c}/{\xi _c}$, 
\begin{equation}
\begin{split}
{\rm{(P5)}}\;\;\{ {{\bf{Y}}_c},{\xi _c}\}=& \arg \mathop {\max }\limits_{{\bf{Y}},\xi } {\rm{Tr[}}{\bf{Y}}(\frac{{\sigma _1^2\sigma _2^2}}{N+1}{\bf{I}} + \sigma _2^2P {{\bf{T}}_1})]\\
&{\rm{s}}{\rm{.t}}{\rm{.}}\;{\rm{Tr[}}{\bf{Y}}(\frac{{\sigma _1^2\sigma _2^2}}{N+1}{\bf{I}} + \sigma _1^2P {{\bf{T}}_2})] = 1,\\
&\qquad{\rm{Tr(}}{\bf{Y}}{{\bf{E}}_i}) = \xi ,\;\;i = 1,2,...,N+1,\qquad\\
&\qquad\;{\bf{Y}} \succeq {\bf{0}},\;\xi  > 0.
\end{split}
\end{equation}
The main idea is still by using the CCT. The details of deriving $({\rm{P}}5)$ are thus omiited for brevity. By introducing a weight $\lambda \in [0,1]$, we can construct a new covariance matrix, i.e.,  
\begin{equation}\label{19}
{{\bf{Z}}}(\lambda) = {\lambda}{{\bf{Z}}_c} + (1 - {\lambda}){{\bf{Z}}_m}.
\end{equation}
Next, we perform the GRP given in (\ref{grp}) to extract a PS vector ${{\bf{v}}}(\lambda)$ from ${{\bf{Z}}}(\lambda)$. By varying the weight $\lambda$ within $[0,1]$, we can obtain a set of ${\bf{v}}(\lambda)$ and then find the optimal one which leads to the maximum secrecy rate. Specifically, for each $\lambda$, we need to obtain its corresponding secrecy rate $R(\lambda)$ for comparison by optimizing the power allocation factor, i.e.,
\begin{subequations}
\begin{align}
{\rm{(P6)}}\;\;R(\lambda) =  &\mathop {\max }\limits_{\alpha } {\left\lceil {\log \frac{{\sigma _1^2\sigma _2^2{\rm{ + }}\sigma _2^2\alpha |{\bf{m}}_1^H{\rm{diag[}}{\bf{v}}{{(\lambda )}^*}]{\bf{g}} + {h_1}{|^2}}}{{\sigma _1^2\sigma _2^2{\rm{ + }}\sigma _1^2\alpha |{\bf{m}}_2^H{\rm{diag[}}{\bf{v}}{{(\lambda )}^*}]{\bf{g}} + {h_2}{|^2}}}} \right\rceil ^ + }\label{6a}\\
&{\rm{s}}{\rm{.t}}{\rm{.,}}\;\log \left( {1 + \frac{{(P-\alpha )|{\bf{m}}_2^H{\rm{diag[}}{\bf{v}}{{(\lambda )}^*}]{\bf{g}} + {h_2}{|^2}}}{{\sigma _2^2 + \alpha |{\bf{m}}_2^H{\rm{diag[}}{\bf{v}}{{(\lambda )}^*}]{\bf{g}} + {h_2}{|^2}}}} \right) \ge {r_m},\qquad\label{6b}\\
&\qquad \alpha  \ge 0,\;P - \alpha  \ge 0,\;\lambda  \ge 0,\;1 - \lambda  \ge 0.
\end{align}
\end{subequations}
where the QoMS constraint for user 1 is discarded due to (\ref{cond}). It is easy to find that the optimal solution to ${\rm{(P6)}}$, denoted as ${\alpha _{{\rm{opt}}}}$, can be derived in closed form. To be specific, since the objective function in (\ref{6a}) monotonically increases with $\alpha$ given the condition (\ref{cond}), ${\alpha _{{\rm{opt}}}}$ should be the largest one within $[0,P]$ that meets constraint (\ref{6b}). In fact, (\ref{6b}) can be rewritten as
\begin{equation}\label{20}
(P - \alpha {2^{{r_m}}})|{\bf{m}}_2^H{\rm{diag}}[{\bf{v}}{(\lambda )^*}]{\bf{g}} + {h_2}{|^2} - ({2^{{r_m}}} - 1)\sigma _2^2 \ge 0.
\end{equation}
It is easy to observe that the left hand side of (\ref{20}) monotonically decreases with $\alpha$. As a result, the optimal allocated power ${\alpha _{{\rm{opt}}}}$ for transmitting the confidential message is given by
\begin{equation}\label{aop}
{\alpha _{{\rm{opt}}}} = {\left\lceil {\min \{ P,\frac{{P|{\bf{m}}_2^H{\rm{diag}}[{\bf{v}}{{(\lambda )}^*}]{\bf{g}} + {h_2}{|^2} - ({2^{{r_m}}} - 1)\sigma _2^2}}{{{2^{{r_m}}}|{\bf{m}}_2^H{\rm{diag}}[{\bf{v}}{{(\lambda )}^*}]{\bf{g}} + {h_2}{|^2}}}\} } \right\rceil ^ + }.
\end{equation}
It is noted from (\ref{aop}) that the optimal transmit power allocated to the confidential message depends critically on the IRS passive beamforming ${\bf{v}}$. More detailed analysis on its impact will be presented in Section V. By substituting (\ref{aop}) into ${\rm{(P6)}}$, we can find the optimal weight $\lambda _{\rm{opt}}$ via the one dimensional search over $\lambda$ within $[0,1]$ and the corresponding PS design, i.e., ${\bf{v}}(\lambda _{\rm{opt}})$, as well as the corresponding power allocation factors ${\alpha _{{\rm{opt}}}}$ and $\beta _{\rm{opt}}=P-{\alpha _{{\rm{opt}}}}$. 
\section{General Multi-User System}
In previous sections, we only discuss the proposed solution to (P1) in the case of $K = 2$. In this section, we consider the general but more challenging multi-user system with $K > 2$ and also analyze the performance gap to the globally optimal solutions. 
\subsection{CCT-based Algorithm for the Multi-User Case}
Similarly as before, we fix the variable ${R_m}$ as a constant ${r_m}$ and the secrecy rate region maximization problem reduces to the maximization of $R_c$, i.e.,
\begin{subequations}
\begin{align}
({\rm{P}}7):\;{{\bar R}_c}({r_m}) &= \mathop {\max }\limits_{\alpha ,\beta ,{\bf{v}}} \mathop {\min }\limits_{k \in \mathcal{E}} \;{\left\lceil {\log \frac{{\sigma _1^2\sigma _k^2{\rm{ + }}\sigma _k^2\alpha |{\bf{m}}_1^H{\rm{diag}}({{\bf{v}}^*}){\bf{g}} + {h_1}{|^2}}}{{\sigma _1^2\sigma _k^2{\rm{ + }}\sigma _1^2\alpha |{\bf{m}}_k^H{\rm{diag}}({{\bf{v}}^*}){\bf{g}} + {h_k}{|^2}}}} \right\rceil ^ + }\\
&{\rm{s}}.{\rm{t}}.\;\;\log \left( {1 + \frac{{\beta |{\bf{m}}_k^H{\rm{diag}}({{\bf{v}}^*}){\bf{g}} + {h_k}{|^2}}}{{\sigma _k^2 + \alpha |{\bf{m}}_k^H{\rm{diag}}({{\bf{v}}^*}){\bf{g}} + {h_k}{|^2}}}} \right) \ge {r_m},\;\forall k\in \mathcal{K},\label{qoss}\\
&\qquad \alpha  + \beta  \le P,\;\;\alpha  \ge 0,\;\;\beta  \ge 0,\;\;\left| {{\bf{v}}(i)} \right| = 1,\;\;i = 1,2, \cdots ,N.
\end{align}
\end{subequations}
By varying the parameter ${r_m}$ within $[0,{r_m^{\rm{up}}}]$, we can find all the boundary points in $({\rm{P}}7)$ and the upper-bound desired multicast rate ${r_{m}^{\rm{up}}}$ in the QoMS constraint can be obtained by solving the following extended SDR from ${\rm{(\hat P3)}}$, i.e., 
\begin{equation}
\begin{split}
{\rm{(P}8):}\;\;&r_m^{\rm{up}}=\mathop {\max }\limits_{{\bf{Z}},\tau } \;\tau \\
&{\rm{s}}{\rm{.t}}{\rm{.}}\;\;{\rm{Tr}}(\frac{P}{{\sigma _k^2}}{\bf{Z}}{{\bf{T}}_k}) \ge {2^\tau } - 1 ,\;\;\forall k\in \mathcal{K},\\
&\qquad\;{\rm{Tr}}({\bf{ZE}}_i)=1,\;\;i = 1,...,N+1,\;\;{\bf{Z}} \succeq {\bf{0}}.
\end{split}
\end{equation}
In the general multi-user case, the necessary condition for ${{\bar R}_c}(0)> 0$ can be expressed as 
\begin{equation}
\frac{{|{\bf{m}}_1^H{\rm{diag}}({{\bf{v}}^*}){\bf{g}} + {h_1}{|^2}}}{{\sigma _1^2}} > \mathop {\max }\limits_k \frac{{|{\bf{m}}_k^H{\rm{diag}}({{\bf{v}}^*}){\bf{g}} + {h_k}{|^2}}}{{\sigma _k^2}}.
\end{equation}
In this case, the QoMS contraint for user 1 in ({\ref{qoss}}) can be discarded and solving $({\rm{P}}7)$ is equivalent to solving $({\rm{P}}9)$,
\begin{equation}
\begin{split}
({\rm{P}}9):\;{2^{{{\overline R }_c}({r_m})}} =& \mathop {\max }\limits_{\alpha ,{\bf{Z}}}\mathop {\min }\limits_{k\in \mathcal{E}} \frac{{{\rm{Tr[}}{\bf{Z}}(\frac{{\sigma _1^2\sigma _k^2}}{N+1}{\bf{I}} + \sigma _k^2\alpha {{\bf{T}}_1})]}}{{ {\rm{Tr[}}{\bf{Z}}(\frac{{\sigma _1^2\sigma _k^2}}{N+1}{\bf{I}} + \sigma _1^2\alpha {{\bf{T}}_k})]}}\\
&{\rm{s}}{\rm{.t}}{\rm{.}}\;\frac{{{\rm{Tr[}}(P -\alpha ){\bf{Z}}{{\bf{T}}_k}]}}{{{\rm{Tr[}}{\bf{Z}}(\frac{{\sigma _k^2}}{N+1}{\bf{I}} + \alpha {{\bf{T}}_k})]}} \ge {2^{{r_m}}} - 1,\;\;k\in \mathcal{E},\\
&\qquad{\rm{Tr(}}{\bf{Z}}{{\bf{E}}_i}) = 1,\;i = 1,2,...,N+1,\;\quad\\
&\qquad{\rm{rank}}({\bf{Z}})=1,\;{\bf{Z}} \succeq {\bf{0}},\;\alpha  \ge 0,\;P - \alpha  \ge 0.
\end{split}
\end{equation}
Due to the minimum operator in the objective function of $({\rm{P}}9)$, we need to further rewrite it as a more tractable form before applying the CCT, i.e.,
\begin{equation}\label{obj}
\mathop {\min }\limits_{k \in {\cal E}} \frac{{{\rm{Tr}}[{\bf{Z}}(\frac{{\sigma _1^2\sigma _k^2}}{{N + 1}}{\bf{I}} + \sigma _k^2\alpha {{\bf{T}}_1})]}}{{{\rm{Tr}}[{\bf{Z}}(\frac{{\sigma _1^2\sigma _k^2}}{{N + 1}}{\bf{I}} + \sigma _1^2\alpha {{\bf{T}}_k})]}} = \mathop {\min }\limits_{k \in {\cal E}} \frac{{{\rm{Tr}}[\frac{{\bf{Y}}}{\xi }(\frac{{\sigma _1^2\sigma _k^2}}{{N + 1}}{\bf{I}} + \sigma _k^2\alpha {{\bf{T}}_1})]}}{{{\rm{Tr}}[\frac{{\bf{Y}}}{\xi }(\frac{{\sigma _1^2\sigma _k^2}}{{N + 1}}{\bf{I}} + \sigma _1^2\alpha {{\bf{T}}_k})]}} = \frac{{{\rm{Tr}}[{\bf{Y}}(\frac{{\sigma _1^2}}{{N + 1}}{\bf{I}} + \alpha {{\bf{T}}_1})]}}{{\mathop {\max }\limits_{k \in {\cal E}} {\rm{Tr}}[{\bf{Y}}(\frac{{\sigma _1^2}}{{N + 1}}{\bf{I}} + \frac{{\sigma _1^2}}{{\sigma _k^2}}\alpha {{\bf{T}}_k})]}}.
\end{equation}
Notice that by properly selecting an auxiliary variable $\xi$, the denominator of the last equation in (\ref{obj}) can be any positive real number. Thus, we can introduce this auxiliary variable $\xi$ along with an additional constraint, i.e., ${\max _{k \in \mathcal{E}}}{\rm{Tr}}[{\bf{Y}}(\frac{{\sigma _1^2}}{{N + 1}}{\bf{I}} + \frac{{\sigma _1^2}}{{\sigma _k^2}}\alpha {{\bf{T}}_k})] = 1$, without loss of optimality. Then, $({\rm{P}}9)$ can be rewritten as an SDR, which is given by
\begin{subequations}
\begin{align}
({\rm{P10}}):\;\;C({r_m})=&\mathop {\max }\limits_{\alpha ,{\bf{Y}},\xi } {\rm{Tr}}[{\bf{Y}}(\frac{{\sigma _1^2}}{{N + 1}}{\bf{I}} + \alpha {{\bf{T}}_1})]\\
&{\rm{s}}.{\rm{t}}.\;{\rm{Tr}}[{\bf{Y}}(\frac{{\sigma _1^2}}{{N + 1}}{\bf{I}} + \frac{{\sigma _1^2}}{{\sigma _k^2}}\alpha {{\bf{T}}_k})] \le 1,\;\forall k \in {\cal E},\label{28b}\\
&\;\;\;\;\;{\rm{Tr}}\left\{ {{\bf{Y}}\left[ {({\rm{P - }}\alpha {2^{{r_m}}}){{\bf{T}}_k} - \frac{{({2^{{r_m}}} - 1)\sigma _k^2}}{{N + 1}}{\bf{I}}} \right]} \right\} \ge 0,\forall k \in {\cal E},\qquad\label{28c} \\
&\;\;\;\;\;\;{\rm{Tr}}({\bf{Y}}{{\bf{E}}_i}) = \xi,\;i = 1,2,...,N + 1,\label{28d}\\
&\;\;\;\quad{\bf{Y}}\succeq {\bf{0}},\;\alpha  \ge 0,\;P - \alpha  \ge 0,\;\xi  > 0.\label{28e}
\end{align}
\end{subequations}
It is easy to observe that ${\rm{(P10)}}$ is an SDP with any fixed $\alpha$; thus, its optimal solution can be derived by performing the following steps. Specifically, let $(\tilde \alpha ,{\bf{\tilde Y}},\tilde \xi )$ be an optimal solution to $({\rm{P10}})$. We aim to find $\tilde \alpha$ by performing a uniform search over $\alpha$ in the following convex optimization problem,
\begin{equation}
\begin{split}
(\hat{\rm{P}}10):\;\;C({r_m},\alpha)=&\mathop {\max }\limits_{{\bf{Y}},\xi } {\rm{Tr}}[{\bf{Y}}(\frac{{\sigma _1^2}}{{N + 1}}{\bf{I}} + \alpha {{\bf{T}}_1})]\\
&{\rm{s}}.{\rm{t}}.\; \rm{(\ref{28b}),\;(\ref{28c}),\;(\ref{28d}),\;{\bf{Y}}\succeq {\bf{0}},\;\xi  > 0.}
\end{split}
\end{equation} 
Obviously, when $\alpha=\tilde\alpha$, we can obtain (${\bf{\tilde Y}},\tilde \xi$). While $(\hat{\rm{P}}10)$ may not be tight, the GRP can be used to obtain a feasible ${\bf{v}}$ from a higher-rank solution. The detailed steps of the generalized CCT-based approach are summarized in Algorithm 1, where $T_{\alpha}$ is the number of uniformly sampling points of $\alpha$. 
\begin{algorithm}
  \caption{Generalized CCT-based Approach for Characterizing the secrecy rate region}
  \begin{algorithmic}[1]
  \Require  $P,\;\{\sigma _k^2\}_{k=1}^K,\;{\bf{g}},\; \{h_k\}_{k=1}^{K},\;\{{\bf{m}}_k^H\}_{k=1}^K,\;T_{\alpha}.$ 
  \State {\bf{Initialization:}} $\mathcal{K}\triangleq \{1,2,...,K\}$, $\mathcal{E}\triangleq \mathcal{K}/\{1\}$, ${{\overline R }_c}({r_m})=0$.  \\
   Compute ${{\bf{T}}_k} = \left[ {\begin{array}{*{20}{c}}
{{\rm{diag}}({\bf{m}}_k^*){\bf{g}}{{\bf{g}}^H}{\rm{diag}}({\bf{m}}_k^{})}&{{\rm{diag}}({\bf{m}}_k^*){\bf{g}}h_k^*}\\
{{h_k}{{\bf{g}}^H}{\rm{diag}}({\bf{m}}_k^{})}&{h_k^*{h_k}}
\end{array}} \right],\;k \in \mathcal{K}.$
  \State Find $r_m^{\rm{up}}$ by solving $({\rm{P}}8)$. 
  \State Vary ${r_m}$ within $[0,{r_m^{\rm{up}}}]$ and obtain ${{\overline R }_c}({r_m})$ through the following steps.
  \State {\bf{for}} $t=1:T_{\alpha}$ {\bf{do}}
  \State \quad $\alpha_t =P(t - 1)/({T_\alpha } - 1)$.
   \State \quad Find ${\bf{Y}}$ and $\xi$ by solving $(\hat{\rm{P}}10)$ with $\alpha=\alpha_t$.
  \State \quad {\bf{if}} $(\hat{\rm{P}}10)$ is feasible {\bf{then}}
  \State \qquad Perform the GRP and obtain a feasible PS vector ${\bf{v}}_t$ by (\ref{grp}).
  \State \qquad Compute ${R_c}(r_m,\alpha _t) = \mathop {\min }\limits_{k \in \mathcal{E}} \;{\left\lceil {\log \frac{{\sigma _1^2\sigma _k^2{\rm{ + }}\sigma _k^2{\alpha _t}|{\bf{m}}_1^H{\rm{diag}}({\bf{v}}_t^*){\bf{g}} + {h_1}{|^2}}}{{\sigma _1^2\sigma _k^2{\rm{ + }}\sigma _1^2{\alpha _t}|{\bf{m}}_k^H{\rm{diag}}({\bf{v}}_t^*){\bf{g}} + {h_k}{|^2}}}} \right\rceil ^ + }$. 
  \State \qquad {\bf{if}} ${R_c}(r_m,\alpha _t)>{{\overline R }_c}({r_m})$ {\bf{then}} ${{\overline R }_c}({r_m})={R_c}(r_m,\alpha _t)$. {\bf{end if}}
      \State \quad {\bf{end if}}
  \State {\bf{end for}}
  \Ensure   ${{\overline R }_c}({r_m})$.
    \end{algorithmic}
\end{algorithm}

It is worth noting that the logarithm of $C({r_m},\tilde\alpha)$ in $(\hat{\rm{P}}10)$ serves as an upper bound on the maximum achievable secrecy rate, which will be used as a benchmark in Section \ref{nu} to evaluate the performance of our proposed algorithms. Given an arbitrary $\tilde\alpha \in [0,P]$, we find via simulations that the achievable secrecy rate can reach the upper bound in most cases, i.e., ${{R}_c}({r_m},\tilde \alpha ) = \log [C({r_m},\tilde \alpha )]$. In this case, $(\hat{\rm{P}}10)$ is tight and any boundary point of the secrecy rate region can be characterized as long as ${R_c}({r_m},\tilde \alpha )$ is found. However, $\tilde \alpha$  may not be exactly found by the one-dimensional search with uniform and finite sampling in Algorithm 1. Besides, $R_c({r_m},\alpha)$ does not monotonically change with $\alpha$ since it is also affected by the optimal PS vector ${\bf{v}}$. Thus, it is challenging to characterize the relationship between the performance gap and the uniform search times in Algorithm 1. In the following Proposition 2, we first characterize the relationship by assuming $(\hat{\rm{P}}10)$ is tight.

\begin{proposition}
Assume that $(\hat{\rm{P}}10)$ is tight, i.e., ${R_c}({r_m},\tilde \alpha )$ is the globally optimal solution to ${\rm{(P7)}}$. $T_\alpha$ and ${{\overline R}_c}({r_m})$ are the number of sampling points and the obtained solution by Algorithm 1, respectively. Let $R_\triangle \triangleq {R_c}({r_m},\tilde \alpha ) - {{\overline R}_c}({r_m})$ be the performance gap. Then, we have
\begin{equation}\label{30}
{R_\triangle} \leq \log \left[ {1 + \frac{{P(N + 1){\rm{Tr(}}{{\bf{T}}_1})}}{{\sigma _1^2({T_\alpha } - 1)}}} \right].
\end{equation}
\end{proposition}
\begin{proof}
See Appendix B.
\end{proof}
Proposition 2 shows that if the solution obtained by steps 7-10 in Algorithm 1 is equal to the performance upper bound, i.e., ${{R}_c}({r_m}, \alpha ) = \log [C({r_m}, \alpha )]$, the performance gap ${R_\triangle }$ decreases to zero when $T_{\alpha}$ is sufficiently large. Due to the difficulty in characterizing the rank profile of the solution to $(\hat{\rm{P}}10)$ analytically, it is important to further characterize the relationship between ${R_\triangle }$ and $T_{\alpha}$ in the case that $(\hat{\rm{P}}10)$ is not tight, i.e., ${{R}_c}({r_m}, \alpha ) \leq \log [C({r_m}, \alpha )]$. 
\begin{corollary}
Let $c =  \arg\mathop{\max }\limits_{t = 1,2,...,{T_\alpha }} {R_c}({r_m},{\alpha _t})$ be the index of the best sampling points in Algorithm 1, i.e., ${R_c}({r_m},{\alpha _c})={\bar R_c}({r_m})$. Then, we have
\begin{equation}\label{38}
{R_\triangle } \leq \log \left[ {1 + \frac{{P(N + 1){\rm{Tr(}}{{\bf{T}}_1})}}{{\sigma _1^2({T_\alpha } - 1)}}} \right] + {\Delta _c}.
\end{equation}
where ${\Delta _c} = \log [C({r_m},{\alpha _c})] - {R_c}({r_m},{\alpha _c})$ is the gap between the performance upper-bound and Algorithm 1. Without knowing ${\Delta _c}$, it ensures the worst-case performance 
\begin{equation}\label{39}
{R_\triangle } \leq \log \left[ {\frac{4}{\pi } + \frac{{4P(N + 1){\rm{Tr}}({{\bf{T}}_1})}}{{\pi \sigma _1^2({T_\alpha } - 1)}}} \right].
\end{equation}
\end{corollary}
\begin{proof}
See Appendix C.
\end{proof} 
To bound the performance gap ${R_\triangle } = {R_c}({r_m},\tilde \alpha ) - {R_c}({r_m},{\alpha _c})$, although it seems that the tightness levels of $(\hat{\rm{P}}10)$ in terms of $\alpha=\tilde \alpha$ and $\alpha = \alpha _c$ are both needed, (\ref{38}) suggests that it suffices to know the latter, i.e., the gap between ${R_c}({r_m}, \alpha_c )$ and $\log [C({r_m}, \alpha_c )]$. When $T_{\alpha}$ is sufficiently large, ${R_\triangle }$ must be less than ${\Delta _c}$. Consequently, after obtaining a solution from Algorithm 1, one can obtain the index $c$ and check ${\Delta _c}$ by calculating $\log [C({r_m},{\alpha _c})]$, thereby evaluating the quality of this solution. Meanwhile, (\ref{39}) guarantees that ${R_\triangle }$  must be no more than $\log (4/\pi )$ even without any knowledge of ${\Delta _c}$.
\subsection{WSCM-based Algorithm for the Multi-User Case}
In this subsection, we extend the WSCM-based algorithm to the general multi-user system. Similarly, we first derive the matrix ${{\bf{Z}}_{m\;}}$ by $(\hat{\rm{P}}8)$ for maximizing the multicast rate. Then we derive the matrix ${{\bf{Z}}_{c\;}}$  for maximizing the secrecy rate by solving the following SDP and setting ${{\bf{Z}}_c} = {{\bf{Y}}_c}/{\xi _c}$, i.e.,
\begin{equation}
\begin{split}
({\rm{P11}}):\;\;\{ {{\bf{Y}}_c},{\xi _c}\}=& \arg\mathop {\max }\limits_{{\bf{Y}},\xi } {\rm{Tr}}[{\bf{Y}}(\frac{{\sigma _1^2}}{{N + 1}}{\bf{I}} + P{{\bf{T}}_1})]\\
&{\rm{s}}.{\rm{t}}.\;{\rm{Tr}}[{\bf{Y}}(\frac{{\sigma _1^2}}{{N + 1}}{\bf{I}} + \frac{{\sigma _1^2}}{{\sigma _k^2}}P{{\bf{T}}_k})] \le 1,\;\forall k \in {\cal E},\\
&\;\;\;\quad{\rm{Tr}}({\bf{Y}}{{\bf{E}}_i}) = \xi ,\;\;i = 1,2,...,N + 1,\\
&\;\;\;\quad{\bf{Y}}\succeq {\bf{0}},\;\alpha  \ge 0,\;P - \alpha  \ge 0,\;\xi  > 0.
\end{split}
\end{equation}
Next, we construct the WSCM by using (\ref{19}). By varying the weight $\lambda$ within $[0,1]$, we need to find the optimal one that leads to the maximum secrecy rate by solving
\begin{subequations}
\begin{align}
{\rm{(P12)}}\;\;R(\lambda) =& \mathop {\max }\limits_{\alpha } \mathop{\min}\limits_{k \in \mathcal{E}} {\left\lceil {\log \frac{{\sigma _1^2\sigma _k^2{\rm{ + }}\sigma _k^2\alpha |{\bf{m}}_1^H{\rm{diag[}}{\bf{v}}{{(\lambda )}^*}]{\bf{g}} + {h_1}{|^2}}}{{\sigma _1^2\sigma _k^2{\rm{ + }}\sigma _1^2\alpha |{\bf{m}}_k^H{\rm{diag[}}{\bf{v}}{{(\lambda )}^*}]{\bf{g}} + {h_k}{|^2}}}} \right\rceil ^ + }\label{12a}\\
&{\rm{s}}{\rm{.t}}{\rm{.,}}\;\log \left( {1 + \frac{{(P-\alpha )|{\bf{m}}_k^H{\rm{diag[}}{\bf{v}}{{(\lambda )}^*}]{\bf{g}} + {h_k}{|^2}}}{{\sigma _k^2 + \alpha |{\bf{m}}_k^H{\rm{diag[}}{\bf{v}}{{(\lambda )}^*}]{\bf{g}} + {h_k}{|^2}}}} \right) \ge {r_m},\;\forall k \in \mathcal{E}.
\label{12b}\\
&\qquad \alpha  \ge 0,\;P - \alpha  \ge 0,\;\lambda  \ge 0,\;1 - \lambda  \ge 0.
\end{align}
\end{subequations}
Let $\tau {\rm{ = }}\arg {\min _k}|{\bf{m}}_k^H{\rm{diag}}[{\bf{v}}{(\lambda )^*}]{\bf{g}} + {h_k}{|^2}/\sigma _k^2.$ Due to the monotonicity of $(P - \alpha )x/(1+ \alpha x)$ w.r.t. $x$ (where $x\in \{|{\bf{m}}_k^H{\rm{diag}}[{\bf{v}}{(\lambda )^*}]{\bf{g}} + {h_k}{|^2}/\sigma _k^2\}_{k=1}^K$), we only need to consider the QoMS constraint of user $\tau$ in (\ref{12b}). Since the objective function of (\ref{12a}) increases with $\alpha$, the optimal $\alpha_{\rm{opt}}$ should be the largest $\alpha$ that satisfies the QoMS constraint for user $\tau$, which is given by 
 \begin{equation}\label{aopt}
{\alpha _{{\rm{opt}}}} = {\left\lceil {\min \{ P,\frac{{P|{\bf{m}}_\tau ^H{\rm{diag}}[{\bf{v}}{{(\lambda )}^*}]{\bf{g}} + {h_\tau}{|^2} - ({2^{{r_m}}} - 1)\sigma _\tau ^2}}{{{2^{{r_m}}}|{\bf{m}}_\tau ^H{\rm{diag}}[{\bf{v}}{{(\lambda )}^*}]{\bf{g}} + {h_\tau}{|^2}}}\} } \right\rceil ^ + }.
\end{equation}
By substituting (\ref{aopt}) into ${\rm{(P12)}}$, we can find $\lambda _{\rm{opt}}$ via a one dimensional search within $[0,1]$. All detailed steps are summarized in Algorithm 2, where $T_{\lambda}$ is the number of uniformly sampling points of $\lambda$.
 \begin{algorithm}
  \caption{WSCM-based Algorithm for Transmit Design with QoMS.}
  \begin{algorithmic}[1]
  \Require  $r_m,\;P,\;\{\sigma _k^2\}_{k=1}^K,\;{\bf{g}},\; \{h_k\}_{k=1}^{K},\;\{{\bf{m}}_k^H\}_{k=1}^K.$ 
  \State {\bf{Initialization:}} $\mathcal{K}\triangleq \{1,2,...,K\}$, $\mathcal{E}\triangleq \mathcal{K}/\{1\}$, ${{\overline R }_c}({r_m})=0$.  \\
   Compute ${{\bf{T}}_k} = \left[ {\begin{array}{*{20}{c}}
{{\rm{diag}}({\bf{m}}_k^*){\bf{g}}{{\bf{g}}^H}{\rm{diag}}({\bf{m}}_k^{})}&{{\rm{diag}}({\bf{m}}_k^*){\bf{g}}h_k^*}\\
{{h_k}{{\bf{g}}^H}{\rm{diag}}({\bf{m}}_k^{})}&{h_k^*{h_k}}
\end{array}} \right],\;k \in \mathcal{K}.$
  \State Find ${{\bf{Z}}_{m\;}}$ by solving $({\rm{P}}8)$.
  \State Find ${{\bf{Z}}_{c\;}}$ by solving $({\rm{P}}11)$.
  \State {\bf{for}} $t=1:T_{\lambda}$ {\bf{do}}
  \State \quad $\lambda_t =(t - 1)/({T_\lambda } - 1)$.
  \State \quad ${{\bf{Z}}}(\lambda_t) = {\lambda_t}{{\bf{Z}}_c} + (1 - {\lambda_t}){{\bf{Z}}_m}.$
  \State \quad Perform the GRP and obtain a feasible PS vector ${\bf{v}}_t$ by (\ref{grp}).
  \State \quad $\tau {\rm{ = }}\arg {\min _k}|{\bf{m}}_k^H{\rm{diag}}[{\bf{v}}{(\lambda_t)^*}]{\bf{g}} + {h_k}{|^2}/\sigma _k^2,$ 
  \State \quad ${\rho _t} = |{\bf{m}}_\tau ^H{\rm{diag}}[{\bf{v}}{({\lambda _t})^*}]{\bf{g}} + {h_\tau }{|^2}/\sigma _\tau ^2.$
   \State \quad ${\alpha _t} = {\left\lceil {\min \{ P,\frac{P}{{{2^{{r_m}}}}} - \frac{1}{{{\rho _t}}} + \frac{1}{{{2^{{r_m}}}{\rho _t}}}\} } \right\rceil ^ + }.$
  \State \quad Compute ${R_c}(t) = \mathop {\min }\limits_{k \in \mathcal{E}} \;{\left\lceil {\log \frac{{\sigma _1^2\sigma _k^2{\rm{ + }}\sigma _k^2{\alpha _t}|{\bf{m}}_1^H{\rm{diag}}({\bf{v}}_t^*){\bf{g}} + {h_1}{|^2}}}{{\sigma _1^2\sigma _k^2{\rm{ + }}\sigma _1^2{\alpha _t}|{\bf{m}}_k^H{\rm{diag}}({\bf{v}}_t^*){\bf{g}} + {h_k}{|^2}}}} \right\rceil ^ + }$. 
  \State \qquad {\bf{if}} ${R_c}(t)>{{\overline R }_c}({r_m})$ {\bf{then}} ${{\overline R }_c}({r_m})={R_c}(t)$. {\bf{end if}}
      \State \quad {\bf{end if}}
  \State {\bf{end for}}
  \Ensure   ${{\overline R }_c}({r_m})$.
    \end{algorithmic}
\end{algorithm}

\subsection{Complexity Analysis}
In this subsection, we give the computational complexity of the proposed algorithms in finding a boundary point of the secrecy rate region. The complexity of both the CCT-based and  WSCM-based algorithms depend on the line search precision. We assume that a standard interior-point method is used to solve the SDPs required in the proposed algorithms and $T_{g}$ times of randomizations are used in GRP to obtain a feasible solution. We use ${A_n^m}$ to represent the computational complexity of the $m$th part in Algorithm $n$.  In Algorithm 1,  $({\rm{P}}8)$ is subjected to one linear matrix inequality (LMI) constraint of size $N+1$, $K-1$ LMI constraints of size 1, and $N+1$ linear matrix equality (LME) constraints of size 1. The computational cost incurred to find an $\epsilon$-optimal solution to $(\hat{\rm{P}}8)$ is on the order of $\ln (1/\epsilon )\sqrt {{\psi }} {\zeta }$, where ${\psi }=2N+K+1$ is the geometric complexity, ${\zeta } = n[{(N + 1)^3} + K + N] + {n^2}[{(N + 1)^2} + K + N] + {n^3}$ is the cost per iteration, and $n = {\mathcal{O}}[{(N + 1)^2} + 1]$. Hence, the complexity of  $({\rm{P}}8)$  is given by (suppressing $\ln (1/\epsilon )$)
\begin{equation}\label{35}
{A_1^1}= \sqrt {2N + K + 1} \left[ {n\left( {{{(N + 1)}^3} + K + N} \right) + {n^2}\left( {{{(N + 1)}^2} + K + N} \right) + {n^3}} \right].
\end{equation}
$(\hat{\rm{P}}10)$ is subjected to one LMI constraint of size $N+1$, $2K-1$ LMI constraints of size 1, and $N+1$ LME constraints of size 1. Thus, the  complexity of $({\rm{P}}10)$  is given by 
\begin{equation}
{A_1^2}=\sqrt {2N + 2K + 1} \left[ {n\left( {{{(N + 1)}^3} + 2K + N} \right) + {n^2}\left( {{{(N + 1)}^2} + 2K + N} \right) + {n^3}} \right].
\end{equation}
The complexity of GRP is on the order of $G={(N + 1)^3} + 8{T_g}{(N + 1)^2}$.
As a result, the complexity of Algorithm 1 for finding a boundary point ${{\overline R }_c}({r_m})$ is on the order of 
\begin{equation}
{A_1}={A_1^1}+{T_\alpha }({A_1^2}+G).
\end{equation}
In Algorithm 2,  the complexity of  $({\rm{P}}8)$  is given by (\ref{35}) and we have ${A_2^1}={A_1^1}$. $({\rm{P}}11)$ is subjected to one LMI constraint of size $N+1$, $K+2$ LMI constraints of size 1, and $N+1$ LME constraints of size 1. Thus, the complexity of $({\rm{P}}11)$  is given by 
\begin{equation}
{A_2^2} = \sqrt {2N + K + 4} \left[ {n\left( {{{(N + 1)}^3} + K + N+3} \right) + {n^2}\left( {{{(N + 1)}^2} + K + N+3} \right) + {n^3}} \right].
\end{equation}
Consequently, the complexity of Algorithm 2 for finding a boundary point is on the order of 
\begin{equation}
{A_2}={A_2^1}+A_2^2+{T_\lambda }G.
\end{equation}
We find via simulation that it only takes $T_\lambda<\mathcal{O}(N^2)$ search times to achieve a comparable performance to the upper bound. Accordingly, the overall complexity of Algorithm 1 and Algorithm 2 are given by $\mathcal{O}(T_\alpha N^{5.5})$ and $\mathcal{O}(N^{5.5})$, respectively.
\section{Theoretical Analysis}
In this section, we analyze the impact of IRS passive beamforming or PSs on the performance of PHY-SI. Comparing to the traditional system without IRS, we aim to figure out how the IRS affects the secrecy rate region or its underlying mechanism. For simplicity, we take the two-user case as an example while the insights revealed can be extended to the general multi-user case similarly. The secrecy rate achieved in the traditional system is given by 
 \begin{equation}\label{srt}
 R_c^{\rm{non}}={\left\lceil {\log \left( {\frac{{1 + \alpha \frac{{|{h_1}{|^2}}}{{\sigma _1^2}}}}{{1 + \alpha\frac{{|{h_2}{|^2}}}{{\sigma _2^2}}}}} \right)} \right\rceil ^ + }.
  \end{equation}
We assume non-zero $R_c^{\rm{non}}$ is achievable; thus, we have
\begin{equation}\label{n41}
\frac{{|{h_1}{|^2}}}{{\sigma _1^2}} > \frac{{|{h_2}{|^2}}}{{\sigma _2^2}}.
\end{equation}
With the assistance of IRS, the secrecy rate achieved in the IRS-assisted system is given by 
\begin{equation}\label{sri}
{R_c^{\rm{IRS}}} = {\left\lceil {\log \left( {\frac{{1 + \alpha \frac{{|{\bf{m}}_1^H{\rm{diag}}({\bf{v}}^*){\bf{g}} + {h_1}{|^2}}}{{\sigma _1^2}}}}{{1 + \alpha \frac{{|{\bf{m}}_2^H{\rm{diag}}({\bf{v}}^*){\bf{g}} + {h_2}{|^2}}}{{\sigma _2^2}}}}} \right)} \right\rceil ^ + }.
\end{equation}
Considering the monotonicity of the logarithm function, we only focus on the term inside it. Define the following function
\begin{equation}
r(\alpha ) = \frac{{1 + \alpha \frac{{|{h_1}{|^2}}}{{\sigma _1^2}}}}{{1 + \alpha \frac{{|{h_2}{|^2}}}{{\sigma _2^2}}}}.
\end{equation}
It is easy to verify that $r(\alpha )$ increases with $\alpha$ under the condition (\ref{n41}), which implies that \emph{a larger $\alpha$ yields better secrecy-rate performance.} Given a desired multicast rate $r_m$, the optimal $\alpha$ (see (\ref{aop})) can be rewritten as 
\begin{equation}\label{n44}
{\alpha ^{{\rm{opt}}}} = \frac{{Px - {c_1}}}{{{c_2}x}},
\end{equation}
where ${c_1} = ({2^{{r_m}}} - 1)\sigma _2^2$, ${c_2} = {2^{{r_m}}}$, and 
\begin{equation}\label{n45}
x = \left\{ {\begin{split}
&{|{h_2}{|^2},\;\;\;{\rm{without}}\;{\rm{IRS}}}\\
&{|{\bf{m}}_2^H{\rm{diag}}({\bf{v}}^*){\bf{g}} + {h_2}{|^2}},\;{\rm{with\;IRS}}
\end{split}}. \right.
\end{equation}
By checking the first-order derivative, it can be verified that ${\alpha ^{{\rm{opt}}}}$ monotonically increases with $x$. Thus, compared to the traditional no-IRS system, if $|{\bf{m}}_2^H{\rm{diag}}[{\bf{v}}{(\lambda )^*}]{\bf{g}} + {h_2}{|^2} > |{h_2}{|^2}$, the power ${\alpha ^{{\rm{opt}}}}$ is larger in the IRS-assisted system. \emph{Hence, the larger channel gain of user 2 brought by the IRS, the larger portion of transmit power $\alpha$ allocated to the confidential message.} 

Next, we discuss how the variation of the channels affects the secrecy rate in the IRS-assisted system. Assume that in both systems, we allocate an equal transmit power $\alpha $ to the confidential message. Let $E_1$ and $E_2$ represent the channel enhancement brought by the IRS for user 1 and user 2, respectively, i.e.,  
\begin{equation}
{E_k} = \frac{{1 + \alpha \frac{{|{\bf{m}}_k^H{\rm{diag}}({\bf{v}}^*){\bf{g}} + {h_k}{|^2}}}{{\sigma _k^2}}}}{{1 + \alpha \frac{{|{h_k}{|^2}}}{{\sigma _k^2}}}},\;\;k \in \{ 1,2\}.
\end{equation}
As such, the secrecy rates $R_c^{\rm{non}}$ and $R_c^{\rm{IRS}}$ (defined in (\ref{srt}) and (\ref{sri}), respectively) follow   
\begin{equation}
{2^{R_c^{\rm{IRS}}}} = \frac{{{E_1}}}{{{E_2}}} \cdot {2^{R_c^{\rm{non}}}}.
\end{equation}
Let $\eta=\frac{{{E_1}}}{{{E_2}}}$ represent the resulting secrecy-rate enhancement with fixed ${\alpha}$. It is obvious that if $\eta>1$, the secrecy rate achieved is larger in the IRS-assisted system. Moreover, a larger $\eta$ yields a larger security performance gain.

Based on the above discussions, it is noted that the maximum achievable secrecy rate in IRS-assisted system is proportional to $\alpha $ and $\eta$. {\it Thus, the security performance in the IRS-assisted system would outperform that in the traditional system if both $\alpha$ and $\eta$ are enlarged by IRS}. To enlarge $\alpha$ under a given QoMS constraint, one can leverage IRS to enhance the channel gain of user 2. However, while IRS improves the channel gain of user 2, it also changes the channel gain of user 1. If the enhancement brought to user 2 is larger than that to user 1, i.e., $E_2 > E_1$, the secrecy-rate performance, in turn, will be impaired. As a result, there exists an interesting trade-off between enhancing the channel of user 1 and that of user 2.
\begin{proposition}
For any given IRS with PS vector ${{\bf{v}}}$, compared to the traditional system, the secrecy-rate performance of the IRS-assisted system is higher if 
 \begin{equation}\label{n48}
 \begin{split}
&|{\bf{m}}_2^H{\rm{diag}}({{\bf{v}}^*}){\bf{g}} + {h_2}| > |{h_2}|\;\;\;{\rm{and}}\\
&|{h_2}|\cdot|{\bf{m}}_1^H{\rm{diag(}}{{\bf{v}}^*}){\bf{g}} + {h_1}| > |{h_1}|\cdot|{\bf{m}}_2^H{\rm{diag(}}{{\bf{v}}^*}){\bf{g}} + {h_2}|.
\end{split}
 \end{equation}
 On the contrary, IRS must impair the secrecy-rate performance if 
 \begin{equation}\label{n49}
 \begin{split}
&|{\bf{m}}_2^H{\rm{diag}}({{\bf{v}}^*}){\bf{g}} + {h_2}| < |{h_2}|\;\;\;{\rm{and}}\\
&|{h_2}|\cdot|{\bf{m}}_1^H{\rm{diag(}}{{\bf{v}}^*}){\bf{g}} + {h_1}| < |{h_1}|\cdot|{\bf{m}}_2^H{\rm{diag(}}{{\bf{v}}^*}){\bf{g}} + {h_2}|.
\end{split}
 \end{equation}
 \end{proposition}
\begin{proof}
We point out that (\ref{n48}) ensures that both $\alpha$ and $\eta$ are enlarged by IRS. The first inequality is easy to prove due to the monotonicity of (\ref{n44}). Next, we prove that the second inequality is a sufficient condition to increase $\eta$. Let ${m_k} = |{\bf{m}}_k^H{\rm{diag(}}{{\bf{v}}^*}){\bf{g}} + {h_k}{|^2}/|{h_k}|^2,\;k \in \{ 1,2\}$. Then, the second inequality is equivalent to $m_1>m_2$. Notice that the first inequality implies that $m_2>1$, then we have the following relation
\begin{equation}
\begin{split}
{E_2} &= \frac{{1 + {m_2} \cdot \alpha \frac{{|{h_2}{|^2}}}{{\sigma _2^2}}}}{{1 + \alpha \frac{{|{h_2}{|^2}}}{{\sigma _2^2}}}}< \frac{{1 + {m_2} \cdot \alpha \frac{{|{h_1}{|^2}}}{{\sigma _1^2}}}}{{1 + \alpha \frac{{|{h_1}{|^2}}}{{\sigma _1^2}}}}< \frac{{1 + {m_1} \cdot \alpha \frac{{|{h_1}{|^2}}}{{\sigma _1^2}}}}{{1 + \alpha \frac{{|{h_1}{|^2}}}{{\sigma _1^2}}}} = {E_1},
\end{split}
\end{equation}
where the first inequality is due to (\ref{41}) and that the function $\frac{{1 + {m_1} \cdot \alpha x}}{{1 + \alpha x}}$ increases with $x$ when $m_2>1$. Hence, the second inequality implies that $\eta$ is enlarged by IRS, i.e., $\eta>1$.  On the contrary, it is easy to verify that (\ref{n49}) implies that both $\alpha$ and $\eta$ are diminished, which completes the proof.
\end{proof}
Given an IRS PS vector ${{\bf{v}}}$, Proposition 3 provides two sufficient conditions to easily evaluate the effects of IRS PS design in terms of secrecy-rate performance. In particular, it reveals that {\it not all} IRS PS designs can improve the secrecy-rate performance.

\section{Numerical Results}\label{nu}
In this section, numerical results are provided to demonstrate the performance of the considered IRS-assisted PHY-SI by our proposed algorithms compared to that without IRS. Assuming a three-dimensional (3D) Cartesian coordinate system, we consider a uniform planar array (UPA) at the IRS and assume all terminals are located in the $x$-$z$ plane without loss of generality. The large-scale path loss is given by
\begin{equation}
L(d) = {L_0} + 10\log {\left( {\frac{d}{{{D_0}}}} \right)^a}.
\end{equation}
where $L_0$ is the path-loss at the reference distance $D_0=1$ m, and $a$ is the path-loss exponent. Due to the scattered obstacles, we set the path-loss exponent between the AP and users as $a_{\rm{BU}}=3.75$. By properly choosing the location of the IRS, we assume that all IRS-involved links experience approximately free-space path loss. As such, the path-loss exponents of IRS involved links are much smaller than that of AP-user link and is set as $a_{\rm{IRS}}=2.2$. The small-scale fading is assumed to be Rician fading, i.e.,
\begin{equation}\label{model}
{\bf{H}} = \sqrt {\frac{\kappa }{{1 + \kappa }}} {{\bf{a}}_r}(\phi _r^{},{\omega _r}){{\bf{a}}_t}{(\phi _t^{},{\omega _t})^H} + \sqrt {\frac{1}{{1 + \kappa }}} {{\bf{H}}^{{\rm{NLoS}}}},
\end{equation}
where $\kappa$  is the Rician factor and the variables $\phi _{t}$ (resp. ${\omega _t}$) and $\phi_r \in [0,2\pi )$ (resp. ${\omega _r}$) are the azimuth (resp. elevation) angles of departure and arrival for the line-of-sight (LoS) component, respectively. In the case of a UPA in the $yz$-plane with $N_t=N_y\times N_z$ elements, the array response vector can be expressed as, 
\begin{equation}
{\bf{a}}(\phi, \omega ) = \frac{1}{{\sqrt {{N_t}} }}{[1,...,{e^{jk{d_a}({n_y}\sin \phi \sin \omega  + {n_z}\cos \omega )}},...,{e^{jk{d_a}({N_y}\sin \phi \sin \omega  + {N_z}\cos \omega )}}]^T},
\end{equation} 
where $1 < {n_y} < {N_y}$, $1 < {n_z} < {N_z}$, $k = 2\pi /\lambda $, $\lambda$ is the wavelength, and $d_a$ is the element spacing. In particular, we have ${{\bf{a}}_t}=1$ for AP-IRS link, ${{\bf{a}}_r}=1$ for IRS-user links, and ${{\bf{a}}_t}={{\bf{a}}_r}=1$ for AP-user links. ${\bf{H}}_{}^{{\rm{NLoS}}}$ is assumed to be Rayleigh fading channel and its entries are randomly generated following complex Gaussian distribution with zero mean and unit variance. Unless otherwise specified, in the simulation, we set $N=10$, $P=1$W, $T_{\alpha}=T_{\lambda}=80$, $L_0=30$dB, $d_a = \lambda /2$, $\kappa =10$, and $\{\sigma _k^2\}_{k=1}^K = -80$dBm.

\begin{figure} 
  \begin{minipage}[t]{0.5\linewidth} 
    \centering 
\includegraphics[width=2.8in]{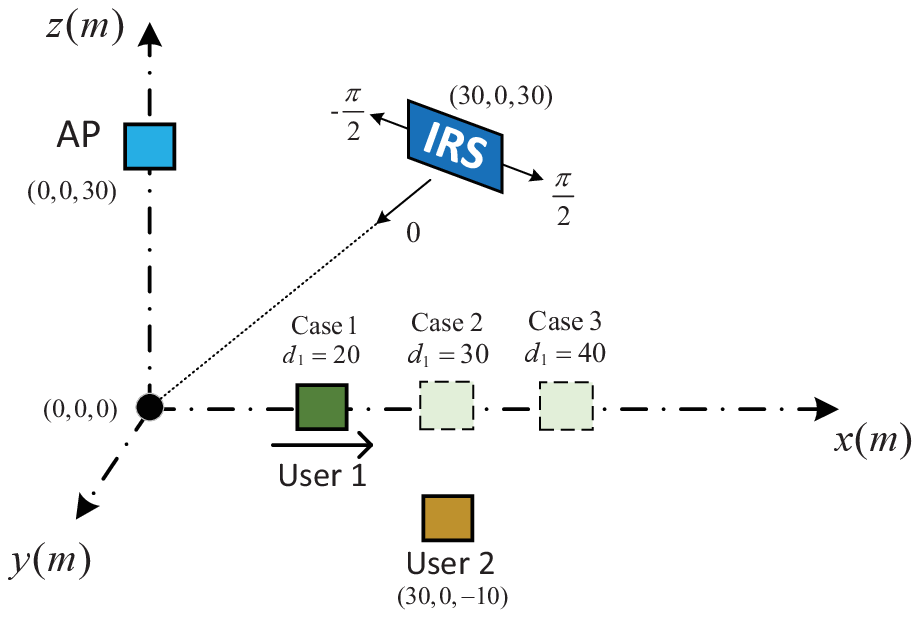}
\caption{Simulation setup of the two-user scenario.}\label{11}
  \end{minipage} 
  \begin{minipage}[t]{0.5\linewidth} 
    \centering 
    \includegraphics[width=3.4in]{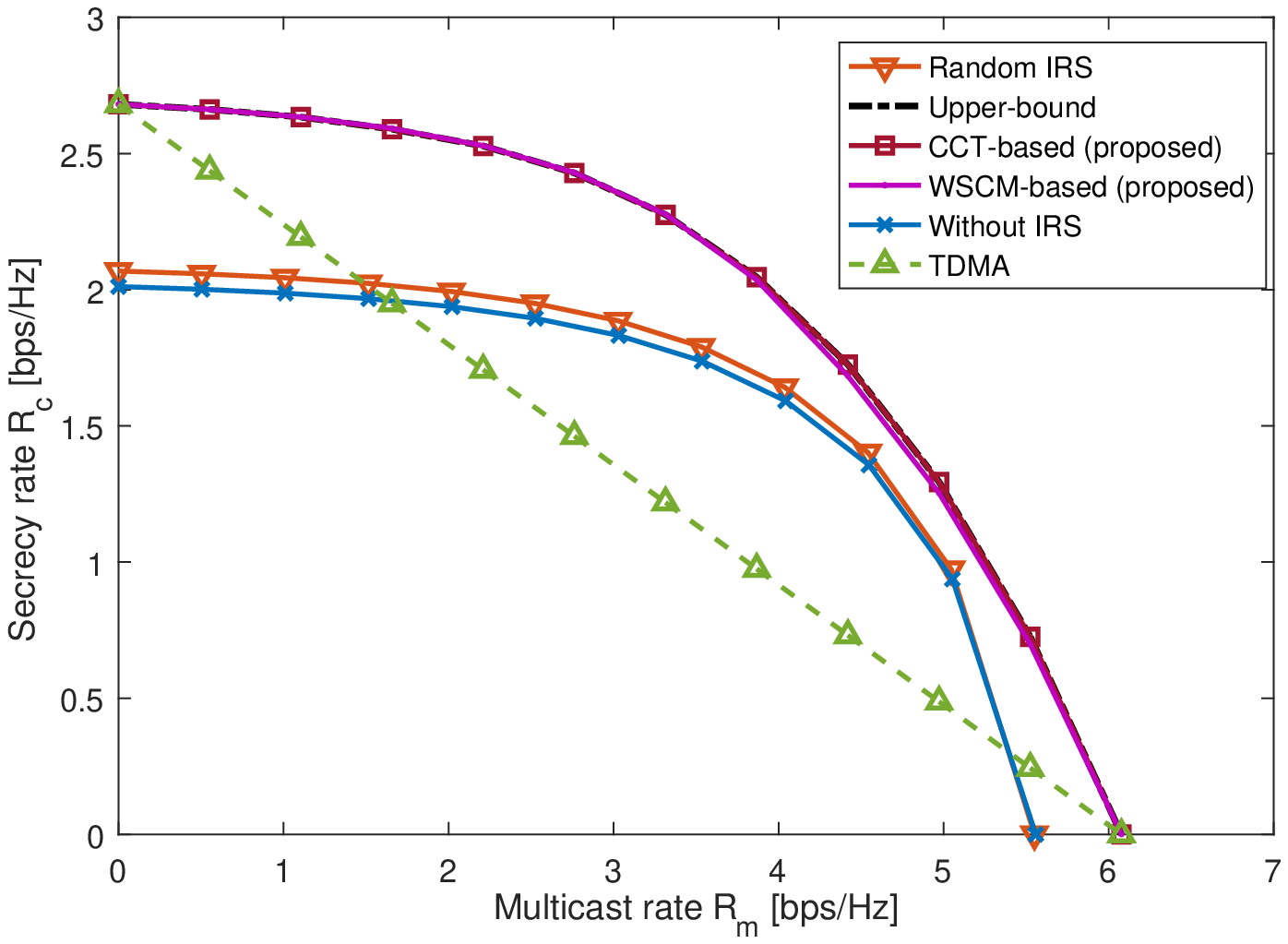} 
    \caption{Secrecy rate regions under different schemes in a two-user scenario.}\label{fd20}
  \end{minipage}%
\vspace{-10pt}
\end{figure}

%
\subsection{Two-User Case, $K=2$}
\begin{table}[t]
\centering
\caption{Simulation Parameters of Channels}
\vspace{-10pt}
\begin{tabular}{|c|c|c|c|c|c|}
\hline
Link & AP-IRS    & AP-user 1 & AP-user 2  & IRS-user 1  & IRS-user 2 \\ \hline
Distance (m)     & 30  &  $\sqrt {{{30}^2} + d_1^2}$   & $50$  & $\sqrt {{{30}^2} + {(30 - {d_1})^2}}$   & 40              \\ \hline
AoA              & $\phi _r^{{\rm{AI}}} = -\frac{\pi}{4}$, $\omega _r^{{\rm{AI}}} = \frac{\pi}{2}$ & $-$ & $-$  &   $-$ &  $-$  \\ \hline
AoD              & $-$  & $-$  &  $-$   & $\begin{array}{*{20}{c}}
{\phi _t^{{\rm{I}}{{\rm{U}}_1}} = \arctan \left( {\frac{{30}}{{30 - {d_1}}}} \right) \!-\! \frac{\pi }{4}}\\
{\omega _t^{{\rm{I}}{{\rm{U}}_1}} = \pi/{2}}
\end{array}$ &              $\begin{array}{*{20}{c}}
{\phi _t^{{\rm{I}}{{\rm{U}}_2}} = \pi /4}\\
{\omega _t^{{\rm{I}}{{\rm{U}}_2}} = \pi /2}
\end{array}$\\ \hline
LoS component& ${{\bf{a}}_r}(\phi _r^{{\rm{AI}}},\omega _r^{{\rm{AI}}})$  &  1  &  1   &  ${{\bf{a}}_t}{(\phi _t^{{\rm{I}}{{\rm{U}}_1}},\omega _t^{{\rm{I}}{{\rm{U}}_1}})^H}$& ${{\bf{a}}_t}{(\phi _t^{{\rm{I}}{{\rm{U}}_2}},\omega _t^{{\rm{I}}{{\rm{U}}_2}})^H}$ \\ \hline
\end{tabular}
\end{table}
First, we consider the two-user scenario. As shown in Fig. \ref{11},  AP, IRS, and user 2 are located at $(0,0,30)$, $(30,0,30)$, and $(30,0,-10)$ respectively. User 1 is located at $(0,0,d_1)$ and we will study the effect of $d_1$ on the performance of PHY-SI later. The center of the IRS is directly pointing to the origin and the parameters of all channels involved are summarized in Table I. We compare our proposed CCT-based and WSCM-based algorithms with the following benchmark schemes:
\begin{itemize}
\item {\bf{Random IRS}}: The PS vector ${\bf{v}}$ is randomly generated, with $|{\bf{v}}(i)|=1,\;i=1,2,...,N.$ Then, the maximum secrecy rate under the QoMS is obtained by using the optimal power allocation given in (\ref{n44}).  
\item {\bf{Upper-bound}}: Given each $r_m$, we find $\alpha _c$ via the CCT-based algorithm and obtain the upper bound $\log [C({r_m},{\alpha _c})]$ by solving $(\hat{\rm{P}}10)$. 
\item {\bf{Time division multiple address (TDMA)}}: The multicast message and confidential message are assigned into orthogonal time slots.
\item {\bf{Without IRS}}: In this scheme, we assume that there is no IRS in the system and the optimal power allocation is given in (\ref{n44}).
\end{itemize}

\begin{figure} 
  \begin{minipage}[t]{0.5\linewidth} 
    \centering 
    \includegraphics[width=3.4in]{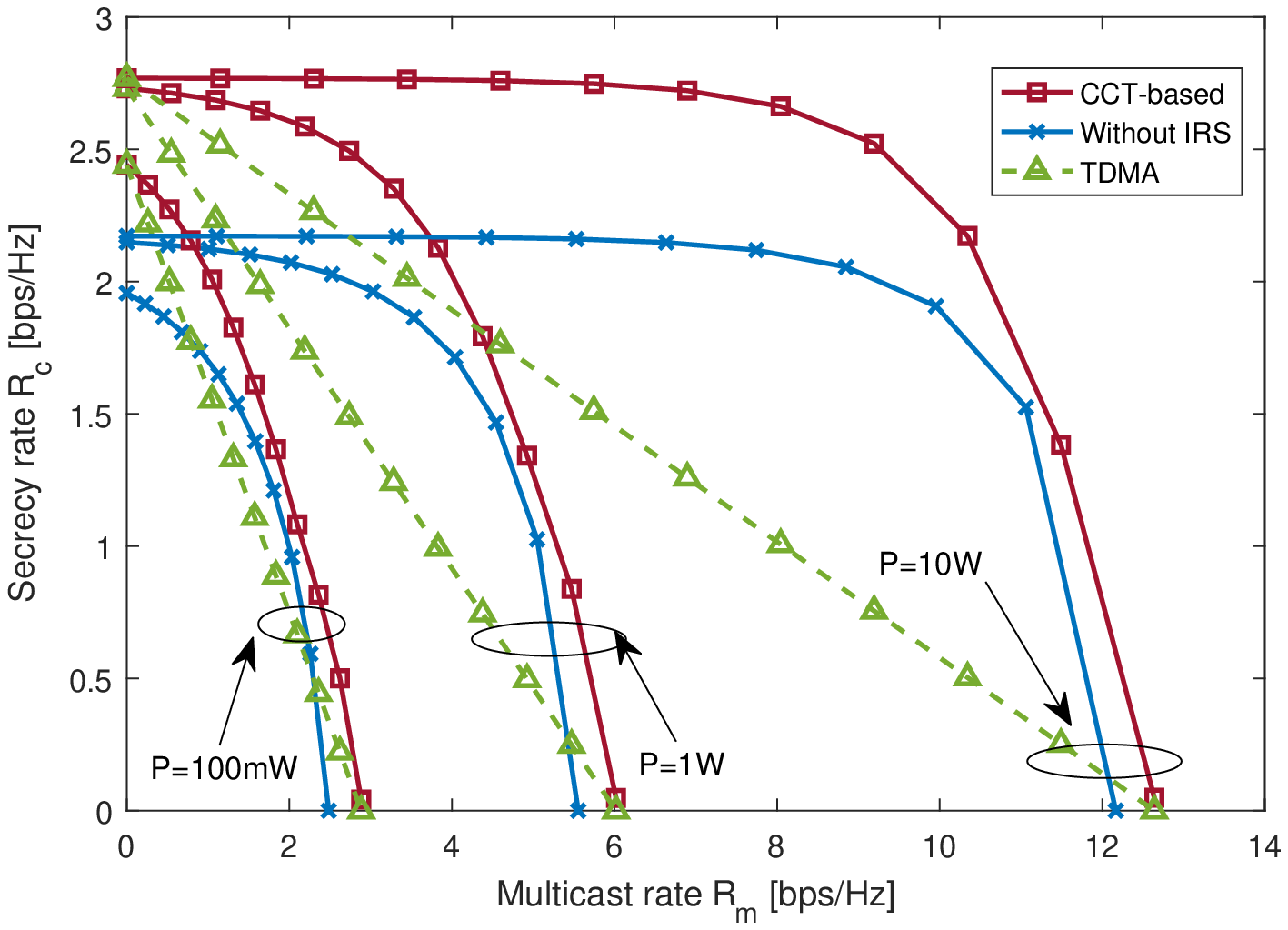} 
   \caption{Secrecy rate regions versus the transmit power.}\label{fpower}
  \end{minipage} 
    \begin{minipage}[t]{0.5\linewidth} 
    \centering 
  \includegraphics[width=3.4in]{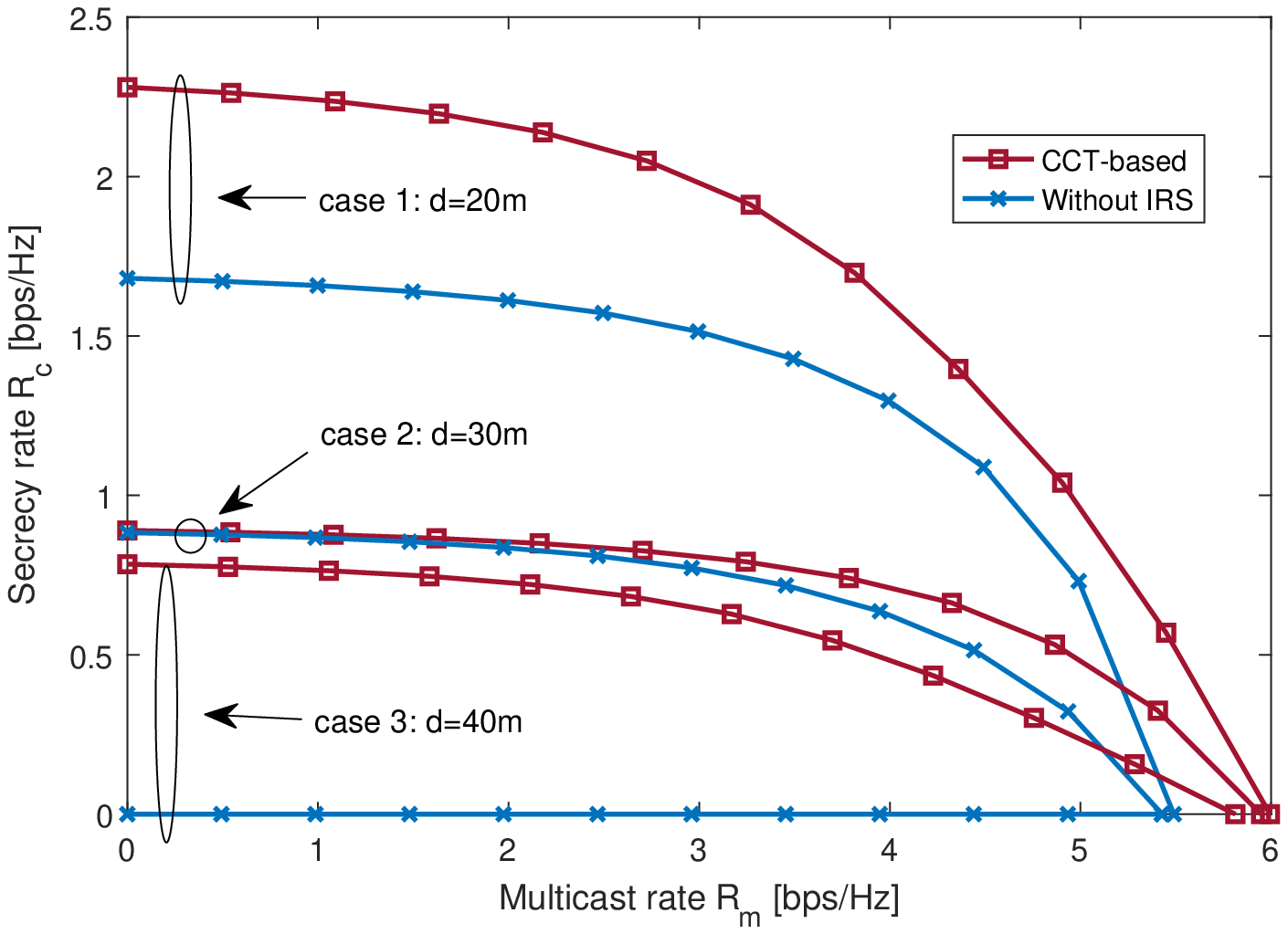}
\caption{Secrecy rate regions with different user 1's locations.}\label{d1}
  \end{minipage}%
\vspace{-10pt}
\end{figure}

%

Fig. \ref{fd20} illustrates the secrecy rate regions achieved by our considered six schemes with $d_1=20$. First of all, let us focus on the performance of our proposed CCT-based and WSCM-based algorithms. It is observed that the CCT-based algorithm can achieve comparable performance to the upper bound, which validates its efficacy. Comparing the WSCM-based algorithm to the CCT-based one, we observe that the performance gap between them is negligible. This indicates that the WSCM-based algorithm can also achieve near-optimal performance, but  with lower computational complexity. Next, we compare the proposed algorithm with other benchmarks. As seen in the figure, the Random IRS scheme brings marginal performance gain over the conventional PHY-SI without IRS, while the schemes with optimized IRS PSs (e.g., using the CCT-based and WSCM-based algorithms) are able to achieve remarkable performance gain over it, especially when the desired multicast rate is low and high. The unsatisfying performance of the random IRS scheme matches our theoretical analysis in Section V, i.e., an arbitrary IRS passive beamforming may not improve the performance of PHY-SI. Finally, for the TDMA scheme, it is observed to yield a much worse performance compared to the proposed schemes. This shows the necessity of spectrum sharing for integrating different services.

Fig. \ref{fpower} compares the secrecy rate regions achieved by different schemes under different transmit powers, namely, $P=100$mW, $1$W and $10$W. It is observed from Fig. \ref{fpower} that a larger transmit power yields a larger secrecy rate region. As the transmit power increases from $100$mW to $10$W, the maximum multicast rate increases faster than the maximum secrecy rate in all schemes. 
By comparing the IRS-assisted schemes with the scheme without IRS, it is observed that the IRS passive beamforming brings pronounced performance gain. In particular, one can observe that the maximum secrecy rate by the CCT-based algorithm with $P=100$mW is even larger than that by traditional non-IRS system with $P=10$W. This implies that using IRS can yield better performance even at a lower transmit power. Furthermore, when the desired multicast rate equals to the half of the maximum multicast rate, the performance gap between the CCT-based algorithm and TDMA attains its maximum, especially when the transmit power is high.
\begin{figure} 
  \begin{minipage}[t]{0.5\linewidth} 
    \centering 
\includegraphics[width=3in]{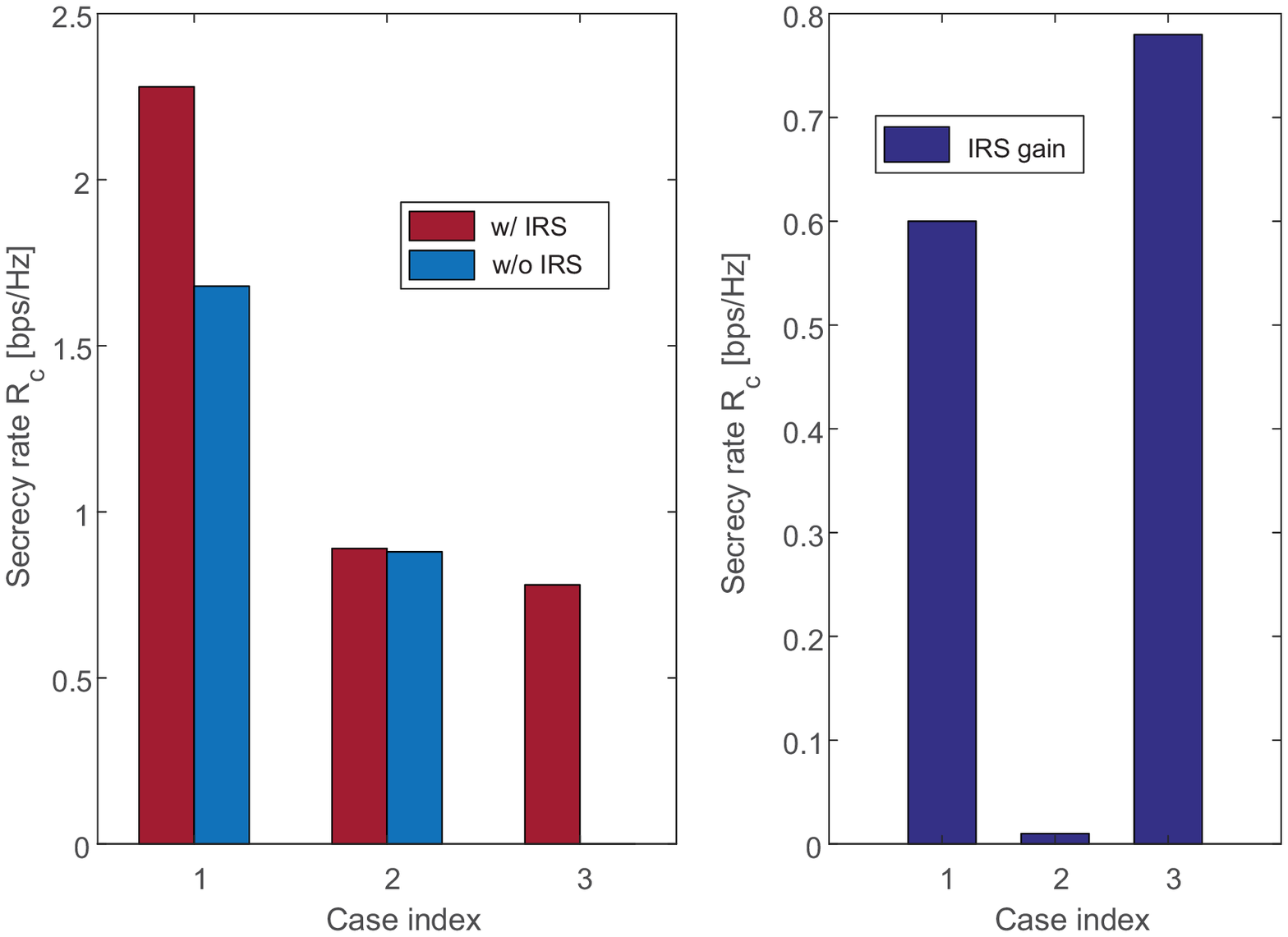}
\caption{Maximum secrecy rates and IRS gains in different cases.}\label{d2}
  \end{minipage} 
    \begin{minipage}[t]{0.5\linewidth} 
    \centering 
\includegraphics[width=3in]{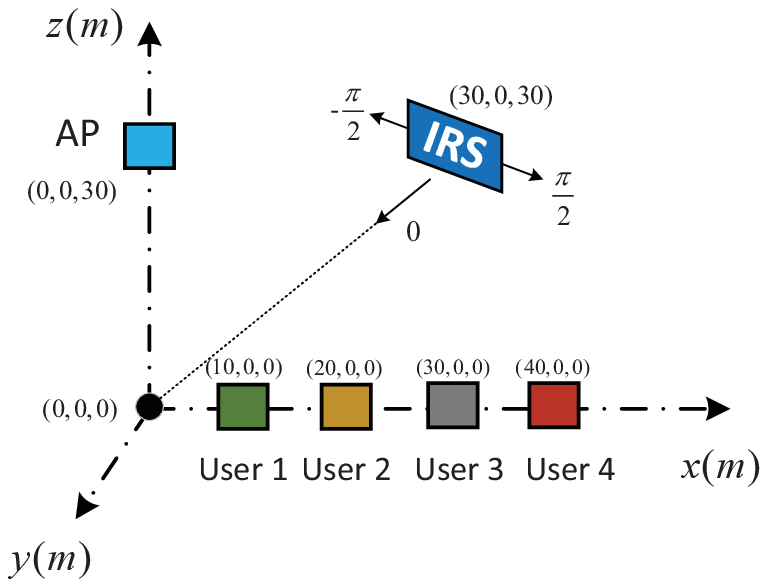}
\caption{Simulation setup of the multi-user system with $K=4$.}\label{22}
  \end{minipage}%
\vspace{-10pt}
\end{figure}


In Fig. \ref{d1},  we evaluate the secrecy rate regions by varying the location of user 1, i.e., $d_1$. Note that we in this example set $\kappa = 10^3$ to reduce the channel randomness. Cases 1, 2 and 3 correspond to the three locations of user 1 with parameters $d_1=20,$ $30$ and $40$, respectively, as depicted in Fig. \ref{11}. Fig. \ref{d1} shows that as user 1 moves away from the AP, i.e., $d_1$ increases, the secrecy rate regions achieved by the CCT-based algorithm and no-IRS system shrink while the reduction in the maximum secrecy rates are much greater than that of the maximum multicast rates. To more clearly explain this observation, Fig. \ref{d2} shows the maximum secrecy rates and the performance gap between the above two schemes. It is observed that in case 3, even though the distance between user 1 and the AP is not smaller than that between user 2 and the AP, the secrecy rate is positive in the IRS-assisted system while it is zero in the traditional system. Besides, their performance gap is rather significant in cases 1 and 3 but is negligible in case 2. This is because in case 2, users 1 and 2 are in the same direction from the IRS, which makes it hard to improve the secrecy rate performance as IRS passive beamforming can potentially improve the channel quality of both users simultaneously. This is consistent with our theoretical analysis in Section V.

\subsection{Multi-User Case}

Now, we consider the multi-user case with $K=4$. As shown in Fig. \ref{22}, we set the locations of the AP and IRS at $(0,0,30)$ and $(30,0,30)$, respectively, while the location of the $k$-th user is located at $(0,0,10k)$, $k=1,2,3,4$. 
We still compare our proposed algorithms to the above mentioned benchmarks. It is worth mentioning that in the traditional scheme, only the user with the shortest distance to AP is able to order confidential information. Due to this fact, user 1 is the legitimate user, while other users are eavesdroppers.

\begin{figure} 
  \begin{minipage}[t]{0.5\linewidth} 
    \centering 
  \includegraphics[width=3.4in]{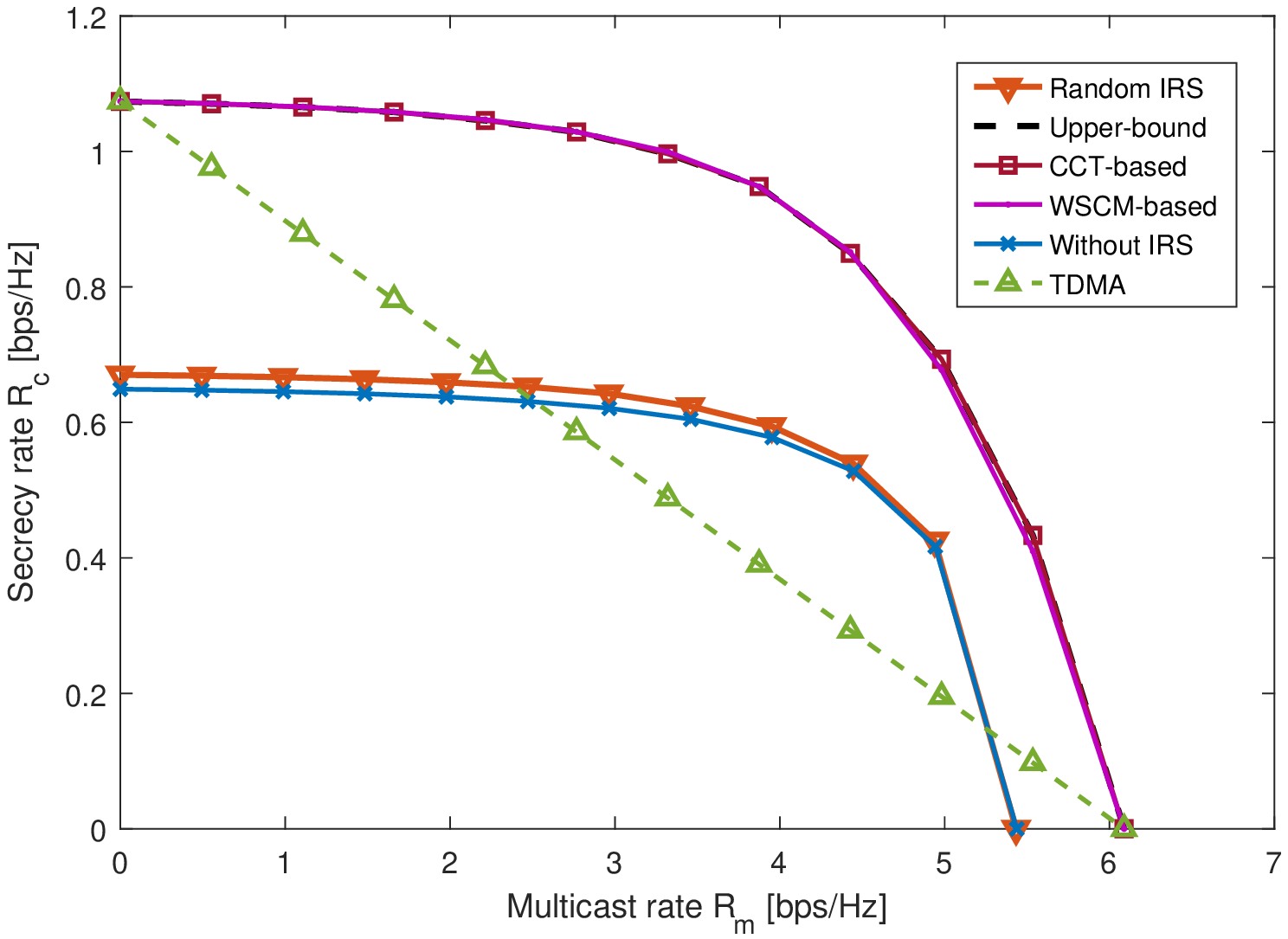}
\caption{Secrecy rate regions by different schemes in the multi-user system with $K=4$.}\label{fm1}
  \end{minipage} 
    \begin{minipage}[t]{0.5\linewidth} 
    \centering 
\includegraphics[width=3.4in]{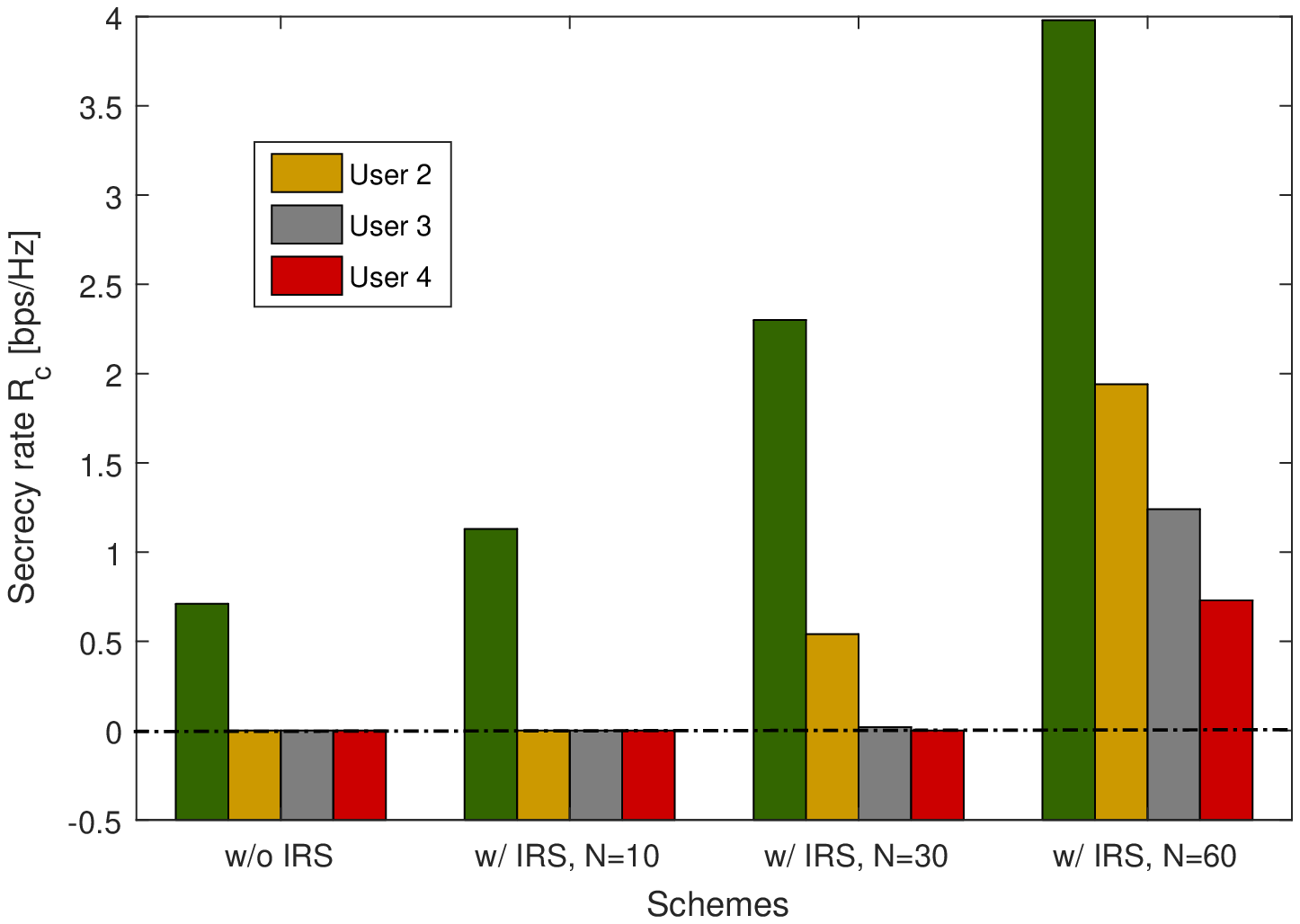}
\caption{Maximum secrecy rates at different users.}\label{d2}
  \end{minipage} 
\vspace{-10pt}
\end{figure}
Fig. \ref{fm1} shows the secrecy rate regions under different schemes in the multi-user scenario. It is observed that the CCT-based algorithm still yields a performance close to the performance upper bound. Moreover, the performance of the random IRS scheme is close to that of the traditional no-IRS system, which implies that only relying on IRS random reflection cannot bring significant performance gain. In contrast, our proposed approaches, i.e., CCT-based and WSCM-based algorithms, strikingly outperform the traditional and TDMA schemes. By comparing Fig. \ref{fm1} with Fig. \ref{fd20}, we observe that the two figures exhibit a similar trend but the secrecy rate regions in Fig. \ref{fm1} are smaller. This is due to the fact that the increase in the number of eavesdroppers is detrimental to the secrecy rate, while the QoMS needs to be satisfied for more users.  

As previously mentioned, users 2-4 cannot achieve positive secrecy rate in the traditional no-IRS system. However, in the presence of an IRS, the effective legitimate and wiretap channels in the system may be significantly changed, making it possible for other users to achieve positive secrecy rates. To verify this fact, we treat users 2-4 as the legitimate users and compare their respective maximum achievable secrecy rates in the traditional and IRS-assisted schemes in Fig. \ref{d2}, where the number of IRS reflecting elements is set to $N=10$, $30$ and $60$.  From Fig. \ref{d2}, it is observed that user 1 can achieve positive secrecy rate in both schemes considered and its secrecy rate monotonically increases with the number of IRS reflecting elements. Moreover, in the IRS-assisted scheme with  $N=30$, when user 2 is treated as a legitimate user, although, although user 1 (eavesdropper) is closest to the AP, user 2 can still achieve a positive secrecy rate thanks to IRS passive beamforming. As $N$ increases to $60$, all users are able to achieve positive secrecy rates, which demonstrates that increasing the number of IRS reflecting elements is an effective means to improve the performance of PHY-SI, allowing more users to order the confidential service.

\section{Conclusions}
We studied a joint power allocation and passive beamforming design for an IRS-aided PHY-SI system. Two customized algorithms, i.e., CCT-based algorithm and WSCM-based algorithm, were proposed to approximately characterize the secrecy rate region. Both optimality and complexity analysis of the proposed algorithms were provided, which show that a complexity-performance trade-off can be flexibly balanced by setting the number of sampling points in the uniform search. Furthermore, to gain essential insights, we performed theoretical analysis to reveal how IRS passive beamforming improves the performance of PHY-SI compared to the traditional system without IRS. Numerical results showed that both proposed algorithms achieve near-optimal performance. It is also revealed that the secrecy rate region can be significantly enlarged by the IRS, and that the number of potential users of the confidential service may be increased thanks to the IRS passive beamforming. Finally, our theoretical analysis was also validated via the numerical results. This work can be extended in several promising directions, e.g., in the more general multi-group multicast MISO/MIMO IRS-assisted systems\cite{gzhou}. Besides, it is also interesting to consider multiple IRSs and investigate their joint passive beamforming design \cite{f1} and cooperative signal reflection\cite{f2}.
 \begin{appendices}
      \section{Proof of Proposition 1}
To determine the feasibility of problem $({\rm{P}}1)$, we need to check whether constraints (\ref{4b}) and (\ref{4c}) are feasible or not. As both ${R_m}$ and ${R_c}$ are non-negative, we only need to verify whether the LHSs of (\ref{4b}) and  (\ref{4c}) are no smaller than zero. It is obvious that the LHS of (\ref{4b}) is always no less than zero for any variables $\alpha$, $\beta$, and ${\bf{v}}$. Hence, it suffices to check the LHS of (\ref{4c}). For (\ref{4c}), it is equivalent to verifying whether there exists an ${\bf{v}}$, such that the following inequality holds,
\begin{equation}\label{40}
|{\bf{m}}_1^H{\rm{diag}}({{\bf{v}}^*}){\bf{g}} + {h_1}{|^2} \ge \mathop {\max }\limits_k \frac{{\sigma _1^2}}{{\sigma _k^2}}|{\bf{m}}_k^H{\rm{diag}}({{\bf{v}}^*}){\bf{g}} + {h_k}{|^2}.
\end{equation}
Notice that ${\bf{m}}_k^H{\rm{diag}}({{\bf{v}}^*}){\bf{g}} = \sum\limits_{i = 1}^N {{e^{j{\theta _i}}}{\bf{m}}_k^*(i){\bf{g}}(i)}$. The maximum absolute value is achieved when the PSs are selected as ${\theta _i}={\phi _c} - \arg [{\bf{m}}_k^*(i){\bf{g}}(i)]$ with an arbitrary ${\phi _c}$. Let ${\bf{v}}_k$ be a PS vector with ${\bf{v}}_k(i) = {e^{ j\arg [{\bf{m}}_k^*(i){\bf{g}}(i)]}}$. If ${\big| {\sum\limits_{i = 1}^N {|{{\bf{m}}_1}(i)||{\bf{g}}(i)|}  + {h_1}} \big|^2} \ge \mathop {\max }\limits_k \frac{{\sigma _1^2}}{{\sigma _k^2}}{\left| {\sum\limits_{i = 1}^N {|{{\bf{m}}_k}(i)||{\bf{g}}(i)|}  + {h_k}} \right|^2}$ holds, there must exist a ${\bf{v}}={\bf{v}}_1$ that satisfies (\ref{40}) due to the following relationship,

\begin{align}
|{\bf{m}}_1^H{\rm{diag}}({\bf{v}}_1^*){\bf{g}} + {h_1}{|^2}&= {\left| {\sum\limits_{i = 1}^N {|{{\bf{m}}_1}(i)||{\bf{g}}(i)| + {h_1}} } \right|^2}\ge \mathop {\max }\limits_k \frac{{\sigma _1^2}}{{\sigma _k^2}}{\left| {\sum\limits_{i = 1}^N {|{{\bf{m}}_k}(i)||{\bf{g}}(i)|}  + {h_k}} \right|^2}\\
&=\mathop {  \max }\limits_k \mathop {\max }\limits_{\bf{v}} \frac{{\sigma _1^2}}{{\sigma _k^2}}|{\bf{m}}_k^H{\rm{diag}}({{\bf{v}}^*}){\bf{g}} + {h_k}{|^2}\ge \mathop {\max }\limits_k \frac{{\sigma _1^2}}{{\sigma _k^2}}|{\bf{m}}_k^H{\rm{diag}}({\bf{v}}_1^*){\bf{g}} + {h_k}{|^2}.\notag
\end{align}
If there exists some $k$ that satisfies ${\big| {\sum\limits_{i = 1}^N {|{{\bf{m}}_1}(i)||{\bf{g}}(i)|}  + {h_1}} \big|^2} \le \frac{{\sigma _1^2|{h_k}{|^2}}}{{\sigma _k^2}}$. We have $\mathop {\max }\limits_k \frac{{\sigma _1^2}}{{\sigma _k^2}}|{\bf{m}}_k^H{\rm{diag}}({{\bf{v}}^*}){\bf{g}} + {h_k}{|^2} > |{\bf{m}}_1^H{\rm{diag}}({{\bf{v}}^*}){\bf{g}} + {h_1}{|^2}$ for any ${\bf{v}}$ due to the following relationship,
\begin{equation}
\begin{split}
\mathop {\max }\limits_k \frac{{\sigma _1^2}}{{\sigma _k^2}}|{\bf{m}}_k^H{\rm{diag}}({{\bf{v}}^*}){\bf{g}} + {h_k}{|^2}&\ge \frac{{\sigma _1^2}}{{\sigma _k^2}}|{\bf{m}}_k^H{\rm{diag}}({{\bf{v}}^*}){\bf{g}} + {h_k}{|^2}> \frac{{\sigma _1^2|{h_k}{|^2}}}{{\sigma _k^2}}\\
&\ge |\sum\limits_{i = 1}^N {|{{\bf{m}}_1}(i)||{\bf{g}}(i)|}  + {h_1}{|^2} \ge |{\bf{m}}_1^H{\rm{diag}}({{\bf{v}}^*}){\bf{g}} + {h_1}{|^2},
\end{split}
\end{equation}
which contradicts (\ref{40}). Thus, the two sufficient conditions given in Proposition 1  are needed to verify the feasibility of $({\rm{P}}1)$.
\section{Proof of Proposition 2}
\begin{figure}[t]
\centering
\includegraphics[width=3in]{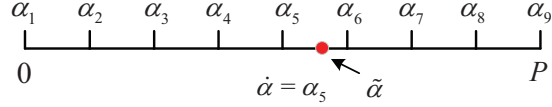}
\caption{An illustration of the uniform line search in Algorithm 1 with $T_{\alpha}=9$.}\label{line}
\vspace{-10pt}
\end{figure}
Let ${\dot \alpha }$ represent the sampling point that is smaller than but closest to $\tilde \alpha$ (see  Fig. \ref{line}). Then, we have
\begin{equation}\label{31}
{R_\triangle } \triangleq {R_c}({r_m},\tilde \alpha ) - {{\bar R}_c}({r_m}) = \log \frac{{C({r_m},\tilde \alpha )}}{{\mathop {\max }\limits_t C({r_m},{\alpha _t})}} \le \log \frac{{C({r_m},\tilde \alpha )}}{{C({r_m},\dot \alpha )}}.
\end{equation}
In particular, ${C({r_m},\tilde \alpha )}$ can be written as the objective function of $(\hat{\rm{P}}10)$ at its optimal solution, i.e., $C({r_m},\tilde \alpha ) = {\rm{Tr}}[{\bf{\tilde Y}}(\frac{{\sigma _1^2}}{{N + 1}}{\bf{I}} + \tilde \alpha {{\bf{T}}_1})]$. Let $\delta  = \tilde \alpha  - \dot \alpha \geq0$; thus, ${C({r_m},\dot \alpha )}$ can be expressed as 
\begin{subequations}\label{32}
\begin{align}
\;\;C({r_m},\dot \alpha)=&\mathop {\max }\limits_{{\bf{Y}},\xi } {\rm{Tr}}\Big[{\bf{Y}}\big(\frac{{\sigma _1^2}}{{N + 1}}{\bf{I}} + (\tilde \alpha-\delta) {{\bf{T}}_1}\big)\Big]\\
&{\rm{s}}.{\rm{t}}.\;{\rm{Tr}}\Big[{\bf{Y}}\big(\frac{{\sigma _1^2}}{{N + 1}}{\bf{I}} + \frac{{\sigma _1^2}}{{\sigma _k^2}}(\tilde \alpha-\delta) {{\bf{T}}_k}\big)\Big] \le 1,\;\forall k \in {\cal E},\label{32b}\\
&\;\;\;\;\;{\rm{Tr}}\left\{ {{\bf{Y}}\left[ {\big({\rm{P - }}(\tilde \alpha-\delta) {2^{{r_m}}}\big){{\bf{T}}_k} - \frac{{({2^{{r_m}}} - 1)\sigma _k^2}}{{N + 1}}{\bf{I}}} \right]} \right\} \ge 0,\forall k \in {\cal E},\qquad\label{32c} \\
&\;\;\;\;\;\;{\rm{Tr}}({\bf{Y}}{{\bf{E}}_i}) = \xi,\;i = 1,2,...,N + 1,\;{\bf{Y}}\succeq {\bf{0}},\;\xi  > 0.\label{32d}
\end{align}
\end{subequations}
Substituting ${\bf{\tilde Y}}$ into (\ref{32b}), we have
\begin{equation}
{\rm{Tr}}\Big[{\bf{\tilde Y}}\big(\frac{{\sigma _1^2}}{{N + 1}}{\bf{I}} + \frac{{\sigma _1^2}}{{\sigma _k^2}}(\tilde \alpha  - \delta ){{\bf{T}}_k}\big)\Big] \leq {\rm{Tr[}}{\bf{\tilde Y}}(\frac{{\sigma _1^2}}{{N + 1}}{\bf{I}} + \frac{{\sigma _1^2}}{{\sigma _k^2}}\tilde \alpha {{\bf{T}}_k})] \le 1,
\end{equation}
where the second inequality is due to the feasibility of ${\bf{\tilde Y}}$ in $(\hat{\rm{P}}10)$. Similarly, we have 
\begin{equation}
 {\rm{Tr}}\left\{ {{\bf{\tilde Y}}\left[ {({\rm{P - }}(\tilde \alpha {\rm{ - }}\delta ){2^{{r_m}}}){{\bf{T}}_k} - \frac{{({2^{{r_m}}} - 1)\sigma _k^2}}{{N + 1}}{\bf{I}}} \right]} \right\} \geq {\rm{Tr}}\left\{ {{\bf{\tilde Y}}\left[ {({\rm{P - }}\tilde \alpha {2^{{r_m}}}){{\bf{T}}_k} - \frac{{({2^{{r_m}}} - 1)\sigma _k^2}}{{N + 1}}{\bf{I}}} \right]} \right\} \ge 0.
\end{equation} 
Thus, $({\bf{\tilde Y}},\tilde \xi)$ is a feasible solution to (\ref{32}). Let $({\bf{\dot Y}},\dot \xi )$ be an optimal solution to (\ref{32}), it holds that
\begin{equation}\label{35}
{\rm{Tr}}[{\bf{\dot Y}}(\frac{{\sigma _1^2}}{{N + 1}}{\bf{I}} + (\tilde \alpha  - \delta ){{\bf{T}}_1})] \ge {\rm{Tr}}[{\bf{\tilde Y}}(\frac{{\sigma _1^2}}{{N + 1}}{\bf{I}} + (\tilde \alpha  - \delta ){{\bf{T}}_1})].
\end{equation}
Based on (\ref{31}) and (\ref{35}), we obtain
\begin{equation}\label{36}
\begin{split}
{R_\triangle } &\leq \log \frac{{C({r_m},\tilde \alpha )}}{{C({r_m},\dot \alpha )}}= \log \frac{{{\rm{Tr}}[{\bf{\tilde Y}}(\frac{{\sigma _1^2}}{{N + 1}}{\bf{I}} + \tilde \alpha {{\bf{T}}_1})]}}{{{\rm{Tr}}[{\bf{\dot Y}}(\frac{{\sigma _1^2}}{{N + 1}}{\bf{I}} + (\tilde \alpha  - \delta ){{\bf{T}}_1})]}}\\
&\le \log \frac{{{\rm{Tr}}[{\bf{\tilde Y}}(\frac{{\sigma _1^2}}{{N + 1}}{\bf{I}} + \tilde \alpha {{\bf{T}}_1})]}}{{{\rm{Tr}}[{\bf{\tilde Y}}(\frac{{\sigma _1^2}}{{N + 1}}{\bf{I}} + (\tilde \alpha  - \delta ){{\bf{T}}_1})]}}\\
& = \log \left( {1 + \frac{{\delta {\rm{Tr(}}{\bf{\tilde Y}}{{\bf{T}}_1})}}{{{\rm{Tr}}[{\bf{\tilde Y}}(\frac{{\sigma _1^2}}{{N + 1}}{\bf{I}} + (\tilde \alpha  - \delta ){{\bf{T}}_1})]}}} \right)\le \log \left( {1 + \frac{{\delta {\rm{Tr(}}{\bf{\tilde Y}}{{\bf{T}}_1})}}{{\tilde \xi \sigma _1^2}}} \right),
\end{split}
\end{equation}
where the last inequality holds due to the fact that ${\rm{Tr}}({\bf{\tilde Y}}{{\bf{E}}_i}) = \tilde \xi,\;\forall i $, and we discard the term $(\tilde \alpha  - \delta ){{\bf{T}}_1}$ as $\tilde \alpha  - \delta = \dot \alpha \geq 0$. Since both ${\bf{\tilde Y}}$ and ${{\bf{T}}_1}$ are positive semi-definite matrices, we have the Cauchy-Schwarz inequality, i.e., ${\rm{Tr(}}{\bf{\tilde Y}}{{\bf{T}}_1}) \le {\rm{Tr(}}{\bf{\tilde Y}}){\rm{Tr(}}{{\bf{T}}_1}) = \tilde \xi (N + 1){\rm{Tr(}}{{\bf{T}}_1})$. In addition, $\delta$ is the gap between $\tilde \alpha$ and $\dot \alpha$, which is no larger than the sampling interval of the linear search, i.e., $\delta  \leq \frac{P}{{{T_\alpha } - 1}}.$ As a result, the performance gap can be further upper bounded as
\begin{equation}\label{37}
{R_\triangle } \leq \log \left( {1 + \frac{{\delta {\rm{Tr(}}{\bf{\tilde Y}}{{\bf{T}}_1})}}{{\tilde \xi \sigma _1^2}}} \right) \le \log \left( {1 + \frac{{\delta (N + 1){\rm{Tr(}}{{\bf{T}}_1})}}{{\sigma _1^2}}} \right) \leq \log \left( {1 + \frac{{P(N + 1){\rm{Tr(}}{{\bf{T}}_1})}}{{\sigma _1^2({T_\alpha } - 1)}}} \right).
\end{equation}

\section{Proof of Corollary 1}
Despite that the difference between (\ref{38}) and (\ref{30}) is small, the proof is not straightforward. By leveraging the inequalities given in (\ref{31}), (\ref{36}) and (\ref{37}), we obtain that 
\begin{equation}
\log \frac{{C({r_m},\tilde \alpha )}}{{C({r_m},{\alpha _c})}}\; = \log \frac{{C({r_m},\tilde \alpha )}}{{\mathop {\max }\limits_t C({r_m},{\alpha _t})}} \le \log \frac{{C({r_m},\tilde \alpha )}}{{C({r_m},\dot \alpha )}} \le \log \left[ {1 + \frac{{P(N + 1){\rm{Tr(}}{{\bf{T}}_1})}}{{\sigma _1^2({T_\alpha } - 1)}}} \right].
\end{equation}
In general cases, an SDR solution to $({\rm{P}}10)$ may not be tight with $\alpha=\alpha_c$ and we have $\log [C({r_m},{\alpha _c})] = {R_c}({r_m},{\alpha _c}) + {\Delta _c}$. Let $\log [C({r_m},\tilde \alpha )] = {R_c}({r_m},\tilde \alpha ) + \tilde \Delta $ with $\tilde \Delta  \ge 0$. Then, it holds that 
\begin{equation}\label{41}
\begin{split}
{R_\triangle }&= {R_c}({r_m},\tilde \alpha ) - {R_c}({r_m},{\alpha _c}) = \log [C({r_m},\tilde \alpha )] - \log [C({r_m},{\alpha _c})] - \tilde \Delta  + {\Delta _c}\\
 &\le \log \left( {1 + \frac{{P(N + 1){\rm{Tr}}({{\bf{T}}_1})}}{{\sigma _1^2({T_\alpha } - 1)}}} \right) - \tilde \Delta  + {\Delta _c} \le \log \left( {1 + \frac{{P(N + 1){\rm{Tr}}({{\bf{T}}_1})}}{{\sigma _1^2({T_\alpha } - 1)}}} \right) + {\Delta _c}.
\end{split}
\end{equation}
It has been shown that the rank-one solution obtained by the GRP guarantees a worst-case performance ratio of $\frac{\pi }{4}$ to the SDR solution\cite{sdr}, i.e., $\frac{{{2^{{R_c}({r_m},{\alpha _c})}}}}{{C({r_m},{\alpha _c})}} \ge \frac{\pi }{4}.$
Due to the relation $C({r_m},{\alpha _c}) = {2^{{R_c}({r_m},{\alpha _c}) + {\Delta _c}}}$, we have
\begin{equation}\label{43}
\frac{{C({r_m},{\alpha _c})}}{{{2^{{R_c}({r_m},{\alpha _c})}}}} = {2^{{\Delta _c}}} \le \frac{4}{\pi }.
\end{equation}
By substituting (\ref{43}) into (\ref{41}), the worst-case performance gap is obtained in (\ref{39}).
\end{appendices}

\end{document}